\documentclass[a4paper,UKenglish,cleveref, autoref, thm-restate]{lipics-v2021}
\title{Disintegration Temporal Logic for Probabilistic
Hyperproperties} %TODO Please add

\author{Mishel Carelli}{CISPA Helmholtz Center for Information Security, Germany}{mishel.carelli@cispa.de}{https://orcid.org/0009-0000-2181-8205}{This
work was supported by the European Research Council (ERC) Grant HYPER
(No. 101055412).}

\author{Bernd Finkbeiner}{CISPA Helmholtz Center for Information Security, Germany,\\
Technical University of Munich, Germany}{finkbeiner@cispa.de}{https://orcid.org/0000-0002-4280-8441}{This
work was supported by the European Research Council (ERC) Grant HYPER
(No. 101055412).}%TODO mandatory, please use full name; only 1 author per \author macro; first two parameters are mandatory, other parameters can be empty. Please provide at least the name of the affiliation and the country. The full address is optional. Use additional curly braces to indicate the correct name splitting when the last name consists of multiple name parts.

\hideLIPIcs

\usepackage{tikz}
\usepackage{ltl}
\usepackage{mathtools}
\usepackage{stmaryrd}
\usepackage{amsmath}

\authorrunning{M. Carelli and B. Finkbeiner} %TODO mandatory. First: Use abbreviated first/middle names. Second (only in severe cases): Use first author plus 'et al.'

\Copyright{Mishel Carelli} %TODO mandatory, please use full first names. LIPIcs license is "CC-BY";  http://creativecommons.org/licenses/by/3.0/

 %TODO mandatory: Please choose ACM 2012 classifications from https://dl.acm.org/ccs/ccs_flat.cfm 

 %TODO mandatory; please add comma-separated list of keywords

\relatedversion{} %optional, e.g. full version hosted on arXiv, HAL, or other respository/website
\acknowledgements{}%optional

\makeatletter
\def\moverlay{\mathpalette\mov@rlay}
\def\mov@rlay#1#2{\leavevmode\vtop{%
		\baselineskip\z@skip \lineskiplimit-\maxdimen
		\ialign{\hfil$\m@th#1##$\hfil\cr#2\crcr}}}
\newcommand{\charfusion}[3][\mathord]{
	#1{\ifx#1\mathop\vphantom{#2}\fi
		\mathpalette\mov@rlay{#2\cr#3}
	}
	\ifx#1\mathop\expandafter\displaylimits\fi}
\makeatother

\renewcommand{\models}{\vDash}

\newcommand{\donotshow}[1]{}

\def\NN{\mathbb{N}}

\def\RR{\mathbb{R}}
\def\PP{\mathbb{P}}
\def\lim{{\varprojlim}}

\def\->{\rightarrow}

\nolinenumbers

\EventEditors{Sayan Bhattacharya, Danupon Nanongkai, Michael Benedikt, and Gabriele Puppis}
\EventNoEds{4}
\EventLongTitle{53rd International Colloquium on Automata, Languages, and Programming (ICALP 2026)}
\EventShortTitle{ICALP 2026}
\EventAcronym{ICALP}
\EventYear{2026}
\EventDate{July 7--10, 2026}
\EventLocation{Royal Holloway, University of London, Egham, United Kingdom}
\EventLogo{}
\SeriesVolume{374}
\ArticleNo{175}

\begin{document}

\maketitle

\begin{abstract}
We introduce \emph{Disintegration Temporal Logic} (DTL), a new probabilistic temporal logic that can express a wide range of probabilistic
hyperproperties, including probabilistic non-interference and perfect
indistinguishability. DTL is based on the notion of \emph{measure disintegration} from probability theory, which allows for conditioning probabilities on a finite or infinite sequence of events occurring during a program execution. This naturally supports reasoning about interacting stochastic systems, where complete executions of one component induce conditional probability distributions over another. We illustrate applications of DTL to systems interacting with stochastic environments, distributional properties of Markov decision processes, and probabilistic automata on infinite words, and discuss its relationship to existing probabilistic logics.

While model checking Markov chains against full DTL is undecidable, we identify two decidable fragments that capture many hyperproperties of interest. The linear fragment admits a polynomial-time model-checking procedure based on linear-algebraic techniques and captures probabilistic information-flow properties such as perfect indistinguishability and history-based probabilistic non-interference. The qualitative fragment admits an automata-theoretic model-checking procedure that extends the standard algorithm for $\mathit{HyperCTL}^*$ with reasoning about bottom strongly connected components.
\end{abstract}

\section{Introduction}

Hyperproperties are system properties that relate multiple execution traces~\cite{ClarksonS10}. They arise naturally in areas such as information-flow security, multi-agent systems, and counterfactual reasoning. There is substantial interest in finding expressive specification logics that capture a large set of relevant hyperproperties and yet have verification problems with acceptable computational cost (cf.~\cite{DBLP:journals/siglog/Finkbeiner23}).

For circuits and other non-probabilistic systems, simple variations of the standard temporal logics, such as HyperLTL~\cite{clarkson2014temporallogicshyperproperties}, the extension of linear-time temporal logic (LTL) with quantification over multiple traces, have proven to be powerful unifying frameworks. Such logics cover a large class of hyperproperties, including notions like non-interference, robustness, and counterfactual causality, that were previously studied in isolation. 
For probabilistic systems, the quest for general specification logics has turned out to be more difficult. A typical example of a probabilistic temporal logic is  HyperPCTL~\cite{abraham2018hyperpctltemporallogicprobabilistic}, which combines standard probabilistic reasoning from logics like PCTL~\cite{DBLP:books/daglib/0020348} with quantification over states. This allows for the specification of restricted versions of probabilistic non-interference, such as that an output distribution must not depend on the initial state. However, logics like HyperPCTL cannot refer to \emph{conditional probabilities}. As a result, such logics cannot express notions like \emph{probabilistic non-interference} in full generality, which is needed for systems that communicate reactively through low-security and high-security channels~\cite{10.5555/2699806.2699811}.

Consider a system that, at every step, receives a low-security input and a high-security input in variables $l$ and $h$ and, based on the previous trace and the current inputs, probabilistically determines the value of a binary output variable $o$. Probabilistic non-interference requires that for any low-level input trace $L$ of length~$n$, any high-level input trace $H$ of length~$n$, and any low-level output trace $O$ of length~$n$, such that the joint probability $\PP(H,L,O) > 0$, the probability of $o$ being true at step $n$ conditioned on $(H,L,O)$, is the same as the probability conditioned only on $(L,O)$. Intuitively, this property expresses that the high-level input trace does not provide additional information about the next output, once the low-level input and output history is fixed.

To express properties like general probabilistic non-interference in a temporal logic, the logic must support both temporal
reasoning and probabilistic reasoning with \emph{conditioning on sequences of events}.
In this paper, we introduce \emph{Disintegration Temporal Logic (DTL)}, which provides such a combination. DTL expresses general probabilistic non-interference as the following formula: 
\begin{equation} \label{nonint}
    \PP\!\Bigl(
  \LTLglobally \bigl(
    \PP(\LTLnext o \mid [h,l,o])
    =
    \PP(\LTLnext o \mid [l,o])
  \bigr)
\Bigr) = 1.
\end{equation}
Here, $\PP(\LTLnext o \mid [h,l,o])$ denotes the probability that $o$ is true at the
next step, conditioned on the history of $h$, $l$, and $o$. Analogously,
$\PP(\LTLnext o \mid [l,o])$ conditions only on the history of $l$ and $o$.
The globally operator $\LTLglobally$ has the usual meaning, that the subformula must hold at every point in time,  and
$\PP(\cdot)=1$ requires that the property holds for almost every trace (up to a set of measure zero). Combining these components, we obtain that the formula~(\ref{nonint}) holds for a system $M$ if the system satisfies probabilistic non-interference.

In addition to conditioning on 
\emph{finite} sequences of events, DTL also allows for conditioning on
\emph{infinite} sequences. This naturally arises in systems consisting of several interacting stochastic components. Consider a system interacting with an external stochastic environment modeled by a Markov chain. DTL can express properties such as \textit{With probability at least $0.1$, the environment produces an execution such that, conditioned on this execution, the probability that the system exhibits the undesirable behavior is greater than $0.4$.} This is expressed by the formula
$$M\models \PP\Bigl( \PP\bigl( \phi_{bad} \mid AP_{env}\bigr) > 0.4 \Bigr) > 0.1.$$

Here, $AP_{env}$ denotes the set of actions produced by the environment.
Note that we write $AP_{env}$ rather than $[AP_{env}]$, corresponding
to conditioning on the full infinite trace rather than on a finite prefix as in formula~(\ref{nonint}).

Beyond these motivating examples, DTL also naturally expresses distributional properties of Markov decision processes and properties of probabilistic automata on infinite words, establishing connections to existing probabilistic verification frameworks.

The mathematical foundation of DTL is the notion of \emph{measure disintegration}.
Disintegration allows us to define conditional probabilities even when conditioning
on events of probability zero.
This is essential in our setting, since individual infinite traces typically have
probability zero, yet we still want to condition on them. The problem of defining conditional probabilities in such cases dates back to
Kolmogorov~\cite{kolomogoroff2013grundbegriffe} and the Borel--Kolmogorov paradox,
which illustrates the subtlety of conditioning on null events.
Measure-theoretic disintegration, established by Rokhlin \cite{Rokhlin1949}, provides a principled resolution of this issue~\cite{book, Chang1997ConditioningAD}.

For full DTL, the model checking problem (determining whether a formula is satisfied by a given Markov
chain) is undecidable. 
However, we identify two decidable fragments that capture many hyperproperties of interest.

The \emph{linear fragment} allows comparing probabilities to $1$ as well as
checking the equality between probabilities.
This fragment can express conditional independence properties, including
probabilistic non-interference and perfect indistinguishability.
The linear fragment admits polynomial-time model checking based on
linear-algebraic reasoning, similar in spirit to algorithms for checking equivalence of labeled Markov chains~\cite{doi:10.1137/0221017}.

The \emph{qualitative fragment} can be viewed as a probabilistic counterpart of the
widely used hyperproperty logic HyperCTL$^*$~\cite{clarkson2014temporallogicshyperproperties}.
In this fragment, probabilities may only be compared to $0$ or $1$, yielding
``soft'' versions of existential and universal quantification.
The qualitative fragment can express, among others, properties of systems interacting with stochastic environments that compare probabilities only to $0$ and $1$. We present a model-checking algorithm for this fragment that combines reasoning about
bottom strongly-connected components with the standard HyperCTL$^*$ model-checking
procedure.
The resulting complexity is a tower of exponentials whose height is linear in the
alternation depth of the formula.

The remainder of the paper is organised as follows:
\begin{itemize}
    \item Section~2 presents preliminaries, introduces measure disintegration, and
    establishes key properties of disintegration on spaces of infinite traces.
    \item Section~3 defines the syntax and semantics of DTL and justifies the semantics
    using measure-theoretic arguments.
    \item Section~4 demonstrates how to express different properties in DTL and
    relates it to existing logics.
    \item Sections~5 and~6 introduce the linear and qualitative fragments and provide
    corresponding model-checking algorithms.
    %\item Section~7 proves the undecidability of general DTL model checking.
    \item Section~7 concludes.
\end{itemize}

\section{Preliminaries}

In this section, we introduce both well-known concepts and new notions that are used throughout the paper.

We begin with notations, followed by basic definitions from probability theory on the space of infinite words. We then introduce the fundamental concept of conditioning via measure disintegration, which forms the basis of our logic. Next, we describe our approach to conditioning on the space of infinite words through the notion of a \textit{cut}, and present several technical results that provide a constructive interpretation of the abstract concept of disintegration in the space of infinite words. Finally, we define Markov chains, Markov decision processes, and various types of automata.

\subsection{Notations}
Let \(\NN\), \(\RR\), and \(\RR_+\) denote the sets of natural numbers, real numbers, and non-negative real numbers, respectively. $\NN_{\omega}$, denotes the set $\NN \cup \{\omega \} 
$, with a linear order, extended from $\NN$, with $\omega>n$, for every $n\in \NN$.

Given an element $s\in X\times Y$, $s_X$ denotes the projection of $s$ on $ Y$.

Given a set $X$ and $n \in \NN$, we write $2^X$ for the set of all subsets of $X$,
$X^n$ for the set of sequences of elements of $X$ of length $n$,
$X^* = \bigcup_{n \in \NN} X^n$,
and $X^\omega$ for the set of all infinite sequences of elements of $X$,
called \emph{traces}. Given a subset $Y\subseteq X$, $\overline{Y}$ denotes the complement of $Y$.

Given a trace $\pi \in X^\omega \cup X^*$ and $n \in \NN$, let $\pi(n)$ denote the $n$-th element of $\pi$ if it exists, and let $\pi[n]$ denote the prefix of $\pi$ of length $\min\{n, |\pi|\}$.

Given $\pi \in X^*$ and $\rho \in X^\omega \cup X^*$, their concatenation is denoted
by $\pi \rho$.
For $Y \subseteq X^*$ and $Z \subseteq X^\omega \cup X^*$, we define their
concatenation element-wise by $YZ = \{ \pi \rho \mid \pi \in Y \text{ and } \rho \in Z \}.$ For $\pi \in X^*$ and $Z\subseteq X^\omega \cup  X^*$, denote $\pi Z = \{ \pi\} Z$.

\subsection{Probability measure on infinite words}
Given a set $X$, a \emph{$\sigma$-algebra} on $X$ is a non-empty collection
$\mathcal{X} \subseteq 2^X$ that is closed under complementation and countable
unions and intersections.
A \emph{measurable space} is a pair $(X,\mathcal{X})$, where $X$ is a set and
$\mathcal{X}$ is a $\sigma$-algebra on $X$.

Given a countable set $X$, the \emph{Cantor $\sigma$-algebra} on $X^\omega$ is the minimal $\sigma$-algebra containing the cylinder sets of the form $\mathsf{Cyl}(\pi) := \{ \pi \rho \mid \rho \in X^\omega \}$, where $\pi \in X^*.$ We always assume the Cantor $\sigma$-algebra on $X^\omega$ and omit mentioning it.

Let $(X,\mathcal{X})$ and $(Y,\mathcal{Y})$ be measurable spaces.
A function $f \colon X \to Y$ is \emph{measurable} if
$f^{-1}(B) \in \mathcal{X}$ for every $B \in \mathcal{Y}$. 
The \emph{product $\sigma$-algebra} on $X \times Y$ is the smallest
$\sigma$-algebra containing all rectangles of the form
$A \times B$ with $A \in \mathcal{X}$ and $B \in \mathcal{Y}$.
We denote it by $\mathcal{X} \otimes \mathcal{Y}$.

A \emph{finite measure} on $(X,\mathcal{X})$ is a function
$\mu \colon \mathcal{X} \to \RR_+$ such that $\mu(\emptyset) = 0$ and $\mu$ is \textit{countably additive}, i.e.
$
\mu (\bigcup_{n \in \NN} A_n) = \sum_{n \in \NN} \mu(A_n)
$
for every sequence of pairwise disjoint sets $A_1, A_2, \dots \in \mathcal{X}$.
A measure $\mu$ is called a \emph{probability measure} if $\mu(X) = 1$.

We denote by $\mathcal{M}(X,\mathcal{X})$ and $\mathcal{P}(X,\mathcal{X})$
the sets of finite measures and probability measures on $(X,\mathcal{X})$,
respectively.

Any countably additive function from the set of cylinder sets of $X^\omega$ to $[0;1]$, that maps $\emptyset$ to $0$, extends uniquely to a probability measure on $X^\omega$ \cite{DBLP:books/daglib/0020348}.

Given a measure $\mu$ on $(X\times Y , \mathcal{X}\otimes\mathcal{Y})$, define a marginal measure $\hat{\mu}_X$ on $(X,\mathcal{X})$ for any $B\in \mathcal{X}$ as $\hat{\mu}_X(B) = \mu(B\times Y)$.

A property on $X$ (for example, the equivalence of two functions or of two sets) is
said to hold \emph{almost everywhere} with respect to a measure $\mu$
(written $\mu$-a.e.), if the set of points where the property fails has
$\mu$-measure zero.

For the real-valued measurable function $f:X\to \RR$, measure $\mu$ and a measurable subset $A\subseteq X$ we denote an integral of $f$ with respect to $\mu$ over $A$ as $\int_A f(x) \mu(dx).$ For the definition of the integral and a more comprehensive introduction to measure theory, we refer to \cite{book, Tao2011AnIT}.

\subsection{Conditioning and Disintegration}

The standard definition of the probability of an event $A$ conditioned on an event $B$ is given by
$$ \PP(A\mid B) = \frac{\PP(A\cap B)}{\PP(B)}. $$
This formula is well-defined only when $\PP(B)\neq 0$. However, in many situations, we would like to condition on events of probability $0$.

Consider, for example, the uniform distribution on the unit square $[0,1]\times [0,1]$. Intuitively, if we condition on the first coordinate being equal to $0.3$, the second coordinate should remain uniformly distributed on $[0,1]$. However, the event that the first coordinate is exactly $0.3$ has probability $0$, so such a conditional distribution is not defined by the standard formula above. To address this issue, probability theory uses the notion of measure disintegration.

Quoting~\cite{book}: "Modern probability theory can be said to begin with the notions of conditioning and disintegration." 

Intuitively, if we have a measure on $X\times Y$, the disintegration of this measure is a well-behaved choice of conditional measures $\mu_x$ on $Y$, indexed by points $x\in X$, called \textit{kernel}.

Let $(X,\mathcal{X})$ and $(Y,\mathcal{Y})$ be measurable spaces.
Formally, a \emph{kernel} from $X$ to $Y$ is a function $\nu \colon X \to \mathcal{M}(Y,\mathcal{Y})$
such that for every $B \in \mathcal{Y}$, the map
$x \mapsto \nu(x)(B)$ is measurable.

Given a measure $\mu\in \mathcal{M}(X,\mathcal{X})$ and a kernel $\nu: X \to \mathcal{M}(Y,\mathcal{Y})$ their \textit{composition} measure $\mu\otimes \nu \in \mathcal{M}(X\times Y, \mathcal{X}\otimes\mathcal{Y})$ is given by 
$$(\mu\otimes \nu) (A\times B) = \int_A  \nu(x)(B) \mu(dx).$$

Intuitively, for each slice corresponding to a fixed value $x\in X$, the measure on $Y$ is given by $\nu(x)$, and these slice measures are then averaged according to $\mu$.

A measurable space $(X,\mathcal{X})$ is called \emph{standard Borel} if there
exists a metric on $X$ that makes a complete,
separable space, such that $\mathcal{X}$ is the associated
Borel $\sigma$-algebra.

We use the fact that $X^\omega$ with Cantor $\sigma$-algebra is standard Borel when $X$
is finite or countable~\cite{Perrin2004InfiniteW}, which suffices to apply the
following result from~\cite{book}.

\begin{theorem}[disintegration] \label{disintegration}
    Given a probability measure $ \mu \in \mathcal{P}(X\times Y, \mathcal{X}\otimes\mathcal{Y})$, where $(Y,\mathcal{Y})$ is standard Borel, there exists a kernel $\nu: X \to \mathcal{P}(Y,\mathcal{Y})$, such that $\mu = \hat\mu_X\otimes \nu$. Such a kernel is unique up to $\hat{\mu}_X$-a.e. equivalence.
\end{theorem}

We call such $\nu$ a disintegration kernel of $\mu$ with respect to $X$ and denote $\nu(x,\cdot)$ as $\mu_x$. Intuitively, $\mu_x$ represents the conditional distribution of $Y$ given $X = x$.

Returning to the example above, let $\mu$ be the uniform distribution on $[0,1]\times [0,1]$. Disintegrating $\mu$ with respect to the first coordinate yields that, for almost every $x\in [0,1]$, the measure $\mu_x$ is the uniform distribution on $[0,1]$. The qualification ``almost every'' is necessary because the disintegration kernel is unique only up to a set of measure zero.

\subsection{Cut and Disintegration in Infinite Words}
Fix a finite set of atomic propositions $AP$. We work with the space $(2^{AP})^\omega$ of infinite words over $2^{AP}$, where at each step every atomic proposition is either true or false.

Observe that there is a natural bijection $(2^{AP})^\omega \cong \prod_{a\in AP} 2^\omega,$. 
We will repeatedly use this correspondence. We also use the natural identification $ 2^n \times 2^\omega \cong 2^\omega, $ given by concatenation.

In this subsection, we consider disintegration in the space of infinite words. To disintegrate, we need to represent $(2^{AP})^\omega$ as a product of two spaces. For that purpose, we introduce a notion of a \textit{cut}.

Informally, a cut can be viewed as a generalized notion of a timestamp. It specifies which part of a trace is conditioned on. Intuitively, a cut separates a trace into a prefix, representing the past, and a suffix, representing the future; disintegration then yields a distribution over future behavior conditioned on the past.

Formally, a \emph{cut} is a function $C \colon AP \to \NN_\omega.$
Given a cut $C$, we define the sets $L(C)$ and $U(C)$ by
$$
L(C) = \prod_{a \in AP} 2^{C(a)},
\qquad
U(C) = \prod_{a \in AP; \ C(a) \neq \omega} 2^{\omega }.
$$

Intuitively, the value $C(a)$ specifies how much of the history of the proposition $a$ is revealed. The factor $2^{C(a)}$ represents the portion of the trace over $a$ that is conditioned on, while the corresponding factor $2^\omega$ in $U(C)$ represents the remaining unconstrained suffix. If $C(a)=\omega$, then the entire history of $a$ is fixed, and therefore no suffix component for $a$ appears in $U(C)$.
\begin{note}
        There is a natural bijection $L(C)\times U(C) \cong  (2^{AP})^\omega$.
\end{note}

Hence, we can view a measure $\mu \in \mathcal{P}((2^{AP})^\omega)$ as a probability
measure on $L(C) \times U(C)$. For each $T' \in L(C)$, we may then consider a
disintegration measure $\mu_{T'}$ as given by
\Cref{disintegration}.

\begin{example}\label{cutex}
Consider the set $AP = \{a_1,a_2\}$ and a cut $C$ such that
$C(a_1)=3$ and $C(a_2)=2$.
Then $L(C)=2^3\times 2^2$, and
$U(C)=2^\omega\times 2^\omega.$
An illustration is given in \Cref{cutfig}.
\end{example}

\begin{figure}[h]
\centering
\begin{tikzpicture}[
    x=0.9cm,
    y=1cm,
    every node/.style={font=\small}
]

% Labels
\node[left] at (0,2) {$a_1$};
\node[left] at (0,1) {$a_2$};

% Timelines
\draw[->] (0,2) -- (8,2);
\draw[->] (0,1) -- (8,1);

% Ticks a1
\foreach \x in {0,...,7}
{
    \draw (\x,2.1) -- (\x,1.9);
}

% Ticks a2
\foreach \x in {0,...,7}
{
    \draw (\x,1.1) -- (\x,0.9);
}

% Time labels
\foreach \x in {0,...,6}
{
    \node[below] at (\x,0.9) {\x};
}

% Conditioned parts
\fill[blue!20] (0,1.8) rectangle (3,2.2);
\fill[blue!20] (0,0.8) rectangle (2,1.2);

% Cut lines
\draw[ thick] (3,1.8) -- (3,2.2);
\draw[ thick] (2,0.8) -- (2,1.2);

% Labels
\node[blue, above] at (1.5,2.2) {conditioned};
\node[above] at (5.2,2.2) {future};

\end{tikzpicture}

\caption{
Illustration of a cut $C$ with
$C(a_1)=3$ and $C(a_2)=2$.
The shaded region represents the part of the trace that is conditioned on.
}
\label{cutfig}
\end{figure}
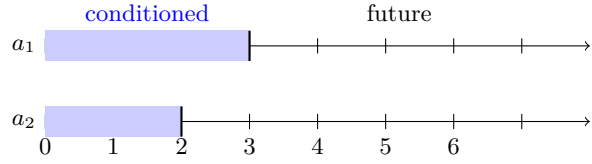

\subsection{Constructive disintegration}

In this subsection, we give a constructive interpretation of the abstract notion of disintegration on the space of infinite words. This result allows us to compute disintegration measures in certain cases and is later used to develop model-checking algorithms.

We first observe that disintegration with respect to a finite prefix coincides
with ordinary conditional probability.
Since the set of prefixes of a fixed finite length is finite, conditioning on a specific prefix amounts to conditioning on an event of positive probability.
In this case, the disintegration measure is given by the classical formula
$$
\PP(X \mid Y) = \frac{\PP(X \cap Y)}{\PP(Y)},
$$
whenever $\PP(Y) > 0$.

    \begin{lemma}\label{constructive_prefix}
Let $\mu \in \mathcal{P}((2^{AP})^\omega)$ and let
$C \colon AP \to \NN$ be a cut that takes only finite values.
Then for every measurable set $B \subseteq U(C)$, for
$\hat{\mu}_{L(C)}$-a.e. $T' \in L(C)$,
$$
\frac{\mu(\{T'\} \times B)}{\mu(\{T'\} \times U(C))}
\;=\;
\mu_{T'}(B),
$$
where the fraction is well-defined $\hat{\mu}_{L(C)}$-a.e.
\end{lemma}

    \begin{proof}
        Let $T'\in L(C)$ be such that $\mu(\{T'\}\times U(C)) > 0$. By the definition of disintegration, for every measurable $B\subseteq U(C)$, we have
        $$\mu(\{ T'\}\times B)= \int_{T'' = T'} \mu_{T''}(B) \hat{\mu}_{L(C)}(dT'')  =$$ $$ = \mu(\{T'\}\times U(C)) \cdot \mu_{T'}(B).$$
        Since $L(C)$ is finite, the set of $T'$, such that $\mu(\{T'\}\times U(C)) = 0$ has a measure zero. Therefore, the stated equality is well-defined and holds $\hat{\mu}_{L(C)}$-a.e.
    \end{proof}

   Next, we show that disintegration with respect to a cut that assigns the value
$\omega$ can be obtained as the limit of disintegrations with respect to its
finite prefixes.

    \begin{theorem}\label{constructive_disintegration}
Let $\mu \in \mathcal{P}((2^{AP})^\omega)$ and let
$C \colon AP \to \NN_\omega$ be a cut.
For every measurable set $B \subseteq U(C)$, for
$\hat{\mu}_{L(C)}$-almost every $T' \in L(C)$,
$$
\frac{\mu(T'_n \times B)}{\mu(T'_n \times U(C))}
\;\xrightarrow[n \to \infty]{}\;
\mu_{T'}(B),
$$
where
$$
T'_n := \{\, T'' \in L(C) \mid T''[n] = T'[n] \,\}.
$$
Every fraction in the sequence is well-defined, and the limit exists
$\hat{\mu}_{L(C)}$-almost everywhere.
\end{theorem}

    \begin{proof}
Fix a measurable set $B \subseteq U(C)$.
For each $n \in \NN$, the set $T'_n \subseteq L(C)$ is finite and measurable.
By Lemma~\ref{constructive_prefix}, for all $n$, for $\hat{\mu}_{L(C)}$-a.e. $T'\in L(C)$, we have

\begin{equation}
    \frac{\mu(T'_n \times B)}{\mu(T'_n \times U(C))}
=
\mu_{T'[n]}(B).
\label{eq_pref}
\end{equation}

By the Conditioning Limits Theorem~\cite[Theorem~9.24]{book}, since $\{ T' \}= \bigcap_{n\in \NN} T'_n $, we obtain
\begin{equation}
\label{eq_lim}
   \mu_{T'[n]}(B) \xrightarrow[n \to \infty]{} \mu_{T'}(B) 
\end{equation}
for $\hat{\mu}_{L(C)}$-almost every $T' \in L(C)$.
The set where \cref{eq_lim} does not hold or \cref{eq_pref} does not hold for some $n$ is a countable union of sets of measure zero, and hence has measure zero.
\end{proof}

\subsection{Modeling the system}
We consider labeled \emph{Markov Decision Processes}.
\begin{definition}
An \emph{MDP} is a tuple $(S, A, AP, l, R, i)$, where $S$, $A$, and $AP$ are finite
sets of states, actions, and atomic propositions, respectively,
$l \colon S \to 2^{AP}$ is a labeling function,
$i \in \mathcal{P}(S)$ is an initial distribution, and
$R \colon S \times A \to \mathcal{P}(S)$ is a transition function.
\end{definition}

A \emph{Markov chain} is an MDP 
$M = (S, A, AP, l, R, i)$ whose action set is a singleton, 
i.e.\ $A = \{\ast\}$. Since the single action is irrelevant, 
we write a Markov chain as $(S, AP, l, P, i)$, omitting $A$, 
where $P(s,s') := R(s,\ast)(s')$.

Given a Markov chain $M = (S, AP, l, P, i)$, we call the sequence of states $\pi \in S^*\cup S^\omega$ a \emph{path} and its corresponding labeling trace $l(\pi)\in {(2^{AP})}^*\cup (2^{AP})^\omega$ is given by $l(\pi)(i) = l(\pi(i)) $, for every $i<|\pi|$. $M$ induces a probability measure $\nu \in \mathcal{P}(S^\omega)$ on the set of infinite paths. For $n\in \NN$ and a finite path $p\in S^n$, the measure of the corresponding cylinder set is given by $$\nu(\pi \in S^\omega \mid \pi[n] = p) = i(p(0) ) \ \prod_{i=0}^{n-2} P(p(i),p(i+1)).$$

Moreover, $M$ induces a probability measure $\mu \in \mathcal{P} ((2^{AP})^\omega)$ on the set of infinite traces. For $n\in \NN$ and a finite trace $t\in (2^{AP})^n$, the measure of the cylinder set is given by $$\mu(T\in (2^{AP})^\omega \mid T[n] = t) = \nu(\pi \in S^\omega \mid l(\pi[n]) =t) = $$ $$ = \sum_{p\in S^n;l(p) = t} i(p(0)) \ \prod_{i=0}^{n-2} P(p(i),p(i+1)).$$

We call a Markov chain $M= (S, AP,l, R, i)$ \textit{precisely labeled}, if for every $s_1,s_2\in S$, if $l(s_1) = l(s_2)$, then $s_1 = s_2$, i.e., labeling function $l: S \to  2^{AP}$ is injective.

Throughout the paper, we use a construction in which one component, $M_1$, “blindly” schedules another component, $M_2$. Intuitively, $M_1$ probabilistically generates a blind scheduler for $M_2$. By \emph{blind} we mean that the scheduler does not depend on the states or transitions of $M_2$, but only on the step index. The probability of generating a particular scheduler is determined solely by $M_1$ and is independent of the behavior of $M_2$. This construction allows us to reason about the behavior of $M_2$ conditioned on the scheduling trace produced by $M_1$. We formalize this intuition using the notion of a \textit{cascade product} $M_1 \otimes M_2$.

\begin{definition}
    Given two MDPs $M_1 = (S_1, A, AP_1,l_1, R_1, i_1)$ and $M_2 = (S_2,2^{AP_1}, AP_2, l_2,R_2,i_2)$, we define their cascade product as $M_1 \otimes M_2 = (S_1\times S_2, A, AP_1\cup AP_2, l, R, i)$, where,
    $$l(s_1,s_2) = l_1(s_1) \cup l_2(s_2),$$
    $$R((s_1,s_2),a)(s_1',s_2') = R_1(s_1,a)(s_1')\cdot R_2(s_2,l_1(s_1'))(s_2'),$$
    $$i(s_1,s_2) = i_1(s_1)\cdot i_2(s_2). $$
\end{definition}

\begin{note}
    If $M_1$ is a Markov chain, then $M_1\otimes M_2$ is a Markov chain as well.
\end{note}

\subsection{Automata}

The algorithm in~\Cref{qualitative_section} uses several types of automata on infinite words.

\begin{definition}[Alternating B\"uchi automaton]\label{def:aba}
An \emph{alternating B\"uchi automaton} (ABA) is a tuple $\mathcal{A} = (Q,\Sigma,\delta,q_0,F),$
where $Q$ is a finite set of states, $\Sigma$ is a finite alphabet,
$q_0 \in B^+(Q)$ is the initial state formula, $F \subseteq Q$ is the set of accepting
states, and the transition function is $\delta \colon Q \times \Sigma \to \mathsf{B}^+(Q),$
where $\mathsf{B}^+(Q)$ is the set of positive Boolean combinations of states. 

Let $w \in \Sigma^\omega$. A \emph{run} of $\mathcal{A}$ on $w$ is a (possibly
infinitely branching) $Q$-labeled tree $r \subseteq \NN^*$ such that:
\begin{itemize}
\item  the set of labels of the children of the root satisfies $q_0$;
\item for every node $u$ of depth $i+1$ labeled by $q$, the set of labels of
the children of $u$ satisfies the formula $\delta(q,w(i))$.
\end{itemize}
A run is \emph{accepting} if every infinite branch visits $F$ infinitely often.
The language $\mathcal{L}(\mathcal{A}) \subseteq \Sigma^\omega$ is the set of
traces that admit an accepting run.
\end{definition}

We call an ABA \emph{existential} (\emph{universal}) if every formula
$\delta(q,\sigma)$ and the initial formula $q_0$ use only disjunctions
( conjunctions).
In this case we may write $\delta(q,\sigma)=Q'$ or $q_0=Q'$ for some
$Q'\subseteq Q$, meaning $\bigvee_{q'\in Q'} q'$ (
$\bigwedge_{q'\in Q'} q'$).

A \emph{reachability} (\emph{safety}) alternating automaton is
defined as ABA above, except that a run is accepting if every
infinite branch that visits $F$ \emph{at least once} (if every
infinite branch \emph{stays in $F$ forever}).

\begin{definition}
A \emph{deterministic Rabin automaton} (DRA) is a tuple $\mathcal{D} = (Q,\Sigma,\delta,q_0,\mathsf{Acc})$, where $Q$ is a finite set of states, $q_0\in Q$ is the initial state, $\delta \colon Q \times \Sigma \to Q$ is a total transition function, and $\mathsf{Acc} = \{(L_i,K_i) \mid i\in I\}, $, where $ L_i,K_i \subseteq Q $,
is a (finite) Rabin acceptance condition.
For $w\in\Sigma^\omega$, the unique run $\rho\in Q^\omega$ is given by
$\rho(n+1)=\delta(\rho(n),w(n))$. It is \emph{accepting} if there exists
$(L_i,K_i)\in\mathsf{Acc}$ such that $\rho$ visits $L_i$ only finitely many times and $K_i$ infinitely many times.
\end{definition}

\section{DTL}

\subsection{Syntax}

Let $AP$ be a finite set of atomic propositions.
The formulas of \emph{Disintegration Temporal Logic} are given by the following grammar, where $\phi$ denotes
state formulas and $\psi$ denotes quantitative formulas:
\begin{align*}
\phi &\coloneqq a \mid \neg \phi \mid \phi \land \phi \mid \psi \leq r
       \mid \LTLnext \phi \mid \phi \LTLuntil \phi, \\
\psi &\coloneqq r \cdot \psi \mid \psi + \psi \mid
       \PP(\phi \mid [A], B).
\end{align*}
Here $a\in AP$ and $A,B\subseteq AP$, and $r\in \RR$.

We use additional temporal operators $\LTLglobally \phi = \neg \LTLeventually \neg \phi$; $\LTLeventually \phi= true \LTLuntil \phi;$ $\phi_1 \LTLr \phi_2 = \neg(\neg \phi_1 \LTLuntil \neg \phi_2)$. Moreover, $\PP(\phi) = \PP(\phi \mid [\emptyset], \emptyset)$.

The \emph{scope} of an occurrence of a probabilistic operator 
$\PP(\phi \mid [A], B)$ is the subformula $\phi$. 
All occurrences of atomic propositions or temporal operators in the scope of a probabilistic operator are said 
to be \emph{bound}. Note that the occurrences of 
atomic propositions from $A \cup B$ in the conditional part 
$[A], B$ are not bound by this operator.

A formula $\phi$ is called \emph{closed} if every occurrence of an atomic proposition or a temporal operator is bound by some probabilistic operator.

\subsection{Semantics}

\subsubsection{Informal explanation.}
The main feature of the logic is the conditional probability operator.
A formula is evaluated relative to a trace $T$, a timestamp $n$, a cut $C$, and
a probability measure $\mu$.

The timestamp represents the current position in the trace; temporal operators
update the timestamp. The cut specifies which part of the trace is conditioned
on, namely the prefix component $L(C)$.

The formula $\PP(\phi \mid [A], B)$ updates the cut from $C$ to $C'$ by assigning
elements of $A$ the current timestamp $n$ and elements of $B$ the value
$\omega$. Intuitively, this corresponds to conditioning on the behavior of
atomic propositions in $A$ up to the current step and on the entire behavior of
atomic propositions in $B$. The operator $\PP(\phi \mid [A], B)$ is interpreted
quantitatively and denotes the probability of satisfying $\phi$, conditioning on $T_{L(C')}$.

\subsubsection{Semantics}

We now define the semantics formally.
For every cut $C$, we fix a disintegration kernel of $\mu$ with respect to
$L(C)$. Later, we show that the semantics do not depend on the particular
choice of this kernel.

Given a trace $T \in (2^{AP})^\omega$, a timestamp $n \in \NN$, a cut
$C \colon AP \to \NN_\omega$, and a measure
$\mu \in \mathcal{P}((2^{AP})^\omega)$, the satisfaction relation
$T,C,n \models_\mu \phi$ is defined inductively as follows:
\begin{align*}
T,C,n &\models_\mu a\text{ iff } a\in T(n), \\
T,C,n &\models_\mu \neg \phi \text{ iff } T,C,n \not\models_\mu\phi, \\
T,C,n &\models_\mu \phi_1 \wedge \phi_2 \text{ iff } T,C,n \models_\mu\phi_1 \text{ and } T,C,n \models_\mu\phi_2, \\
T,C,n &\models_\mu \psi \leq r\text{ iff }  \llbracket T,C,n,\psi \rrbracket_\mu \leq r, \\
T,C,n &\models_\mu \LTLnext \phi \text{ iff } T,C,n+1 \models_\mu \phi, \\
T,C,n &\models_\mu \phi_1 \LTLuntil \phi_2 \text{ iff } \exists i. \ T,C,n+i \models_\mu\phi_2 \text{ and } \forall j<i. \ T,C,n+j \models_\mu\phi_1.
\end{align*}

The formula \(\psi\) is interpreted quantitatively over a trace $T$, a timestamp $n\in \NN$, a cut $C $ and a measure $\mu$, with its semantics taking values in \(\mathbb{R}\).
\begin{align*}
\llbracket T,C,n, r\cdot \psi \rrbracket_\mu &= r \cdot \llbracket T,C,n,\psi \rrbracket_\mu, \\ 
\llbracket T,C,n, \psi_1 + \psi_2 \rrbracket_\mu &= \llbracket T,C,n,  \psi_1 \rrbracket_\mu + \llbracket T,C,n,  \psi_2 \rrbracket_\mu, \\
\llbracket T,C,n, \PP    (\phi \mid [A],B) \rrbracket_\mu &= \mu_{T_{L(C')}}   ( t\in U(C') \mid (T_{L(C')},t),C',n \models_\mu\phi ).
\end{align*}
In the last line, the cut $C':AP\to \NN_\omega$ is defined for $p\in AP$ by
\begin{equation}
    C'(p) = \begin{cases}
    C(p), &\text{ for }p\notin A\cup B,\\
    max(C(p),n),  & \ \text{for } p\in A,\\
    \omega, &\text{ for } p\in B.
\end{cases}
\label{eq_s'}
\end{equation}

Here, $T_{L(C')}$ denotes the projection of
$T \in (2^{AP})^\omega \cong L(C') \times U(C')$ onto $L(C')$, and
$\mu_{T_{L(C')}}$ is the corresponding disintegration measure obtained from the fixed disintegration kernel.

\begin{note}
Disintegration measures for different prefixes must be obtained from the same
kernel. Otherwise, inconsistent choices may violate the measurability of the
semantic interpretation.
\end{note}

\begin{example}
Consider the cut $C$ with $C(a_1)=3$ and $C(a_2) = 2$ from \Cref{cutex}. Suppose we evaluate the formula $\llbracket T,C,5,\PP(\phi \mid [a_1],a_2)\rrbracket_\mu.$
The cut must then be updated according to \Cref{eq_s'}. For the resulting cut $C'$, we have $C'(a_1)=5,$
and $C'(a_2)=\omega.$
Indeed, the formula requires conditioning on the history of $a_1$ up to the current time and on the entire history of $a_2$.
\end{example}

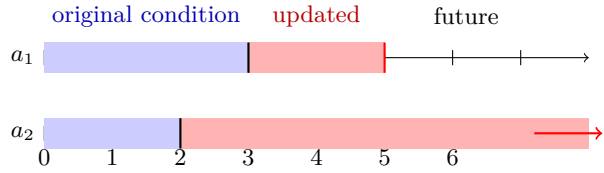
\begin{figure}[h]
\centering
\begin{tikzpicture}[
    x=0.9cm,
    y=1cm,
    every node/.style={font=\small}
]

% Labels
\node[left] at (0,2) {$a_1$};
\node[left] at (0,1) {$a_2$};

% Timelines
\draw[->] (0,2) -- (8,2);
\draw[->] (0,1) -- (8,1);

% Ticks a1
\foreach \x in {0,...,7}
{
    \draw (\x,2.1) -- (\x,1.9);
}

% Ticks a2
\foreach \x in {0,...,7}
{
    \draw (\x,1.1) -- (\x,0.9);
}

% Time labels
\foreach \x in {0,...,6}
{
    \node[below] at (\x,0.9) {\x};
}

% Conditioned parts
\fill[blue!20] (0,1.8) rectangle (3,2.2);
\fill[blue!20] (0,0.8) rectangle (2,1.2);

% Newly conditioned part after update (red)

% a1 : extend from 3 to 5
\fill[red!30] (3,1.8) rectangle (5,2.2);

% a2 : extend from 2 to infinity (depicted by the visible remainder)
\fill[red!30] (2,0.8) rectangle (8,1.2);

% Original cut lines
\draw[ thick] (3,1.8) -- (3,2.2);
\draw[ thick] (2,0.8) -- (2,1.2);

% Updated cut lines
\draw[ thick, red] (5,1.8) -- (5,2.2);

% For a2 the updated cut is omega:
% draw a red arrow indicating continuation to infinity
\draw[red, thick, ->] (7.2,1) -- (8.2,1);
%\node[black, above] at (7.6,1.15) {$\omega$};

% Labels

\node[] at (6.2,2.55) {future};

\node[blue!70!black] at (1.5,2.55) {original condition };
\node[red!70!black] at (4.0,2.55) {updated};

\end{tikzpicture}

\caption{
Illustration of a cut update. The blue region represents the original cut $C$ with $C(a_1)=3$ and $C(a_2)=2$. The red region shows the additional portion required after evaluation of the formula, yielding the updated cut $C'$ with $C'(a_1)=5$ and $C'(a_2)=\omega$.
}
\end{figure}

Given a Markov chain $M$ inducing a measure $\mu$ and a closed DTL formula $\phi$,
we define
$$
M \models \phi
\quad\text{iff}\quad
T, C_0, 0 \models_\mu \phi,
$$
where $C_0(p) = 0$ for all $p \in AP$ and $T$ is an arbitrary trace.
Since $\phi$ is closed, the evaluation does not depend on the choice of $T$.

The \emph{DTL model-checking problem} is to decide, given a Markov chain $M$
and a closed DTL formula $\phi$, whether
$M \models \phi$ holds.

\subsubsection{Correctness and Undecidability}
Recall that disintegration kernels are defined uniquely up to a.e.
equivalence. The following theorem formalizes the correctness of the semantics
and shows that the semantic evaluation is independent of the particular choice
of disintegration kernels.

\begin{theorem}\label{thm:semantic_correctness}
Given DTL formulas $\phi$ and $\psi$, defined by the grammar above, a cut $C \colon AP \to \NN_\omega$, a timestamp
$n \in \NN$, and a probability measure $\mu$, the following statements hold.
\begin{enumerate}
\item
The function
$f \colon (2^{AP})^\omega \to \RR$ defined by
$$
f(T) = \llbracket T, C, n, \psi \rrbracket_\mu
$$
is measurable. Moreover, different choices of disintegration kernels yield
$\mu$-a.e. equivalent functions.

\item
The set
$$
\{\, T \in (2^{AP})^\omega \mid T, C, n \models_\mu \phi \,\}
$$
is measurable. Moreover, different choices of disintegration kernels yield
$\mu$-a.e. equivalent sets.
\end{enumerate}
\end{theorem}
\begin{proof}
    We prove both statements simultaneously by induction on the structure of the
formula.

    \textbf{Base case.}
For atomic formulas $\phi = a$, measurability is immediate, and the evaluation
does not depend on any disintegration kernel.

    \textbf{Induction step (Boolean and temporal operators).}
Assume the induction hypotheses hold for all strict subformulas.

Negation and conjunction correspond to complementation and intersection of
measurable sets, respectively, and therefore preserve measurability. The a.e.-equivalence is preserved under these operations.

The operator $\LTLnext \phi$ defines a measurable set since it shifts the
timestamp and relies on the induction hypothesis.

The formula $\phi_1 \LTLuntil \phi_2$ defines a countable Boolean combination of
measurable sets and is therefore measurable. The a.e.-equivalence is
again preserved.

\textbf{Induction step (Comparison operator).}
For formulas of the form $\psi \le r$, measurability follows from the induction
hypothesis applied to $\psi$. Moreover, if two interpretations of $\psi$ are
$\mu$-a.e. equivalent, then the induced satisfaction sets are also
$\mu$-a.e. equivalent.

    \textbf{Induction step (Quantitative operators).}
By the induction hypothesis, the operators $r \cdot \psi$ and
$\psi_1 + \psi_2$ define linear combinations of measurable functions and are
therefore measurable. The a.e.-equivalence is preserved.

\textbf{Induction step (Probabilistic operator).}
    It remains to consider $\psi = \PP(\phi \mid [A], B)$.
    Let $C'$ be the updated cut defined in~\eqref{eq_s'}.
    
    By the definition, $\llbracket T, C,n,\psi \rrbracket_\mu$ depends only on the projection $T_{L(C')}$. Hence, it is enough to prove the measurability and $\hat{\mu}_{L(C')}$-a.s. equivalence of the function $g: L(C') \to \RR$ defined by
    $$g(T') = \mu_{T'}(t\in U(C') \mid (T',t), C',n \models_\mu \phi ).$$
    
    The characteristic function $\mathbf{1}_{\phi}: L(C')\times U(C') \to \RR$ of the set $\{T\in (2^{AP})^\omega \mid T,C',n \models_\mu \phi \}$ is measurable by the induction hypothesis.   
 
    By the ~\cite[Lemma~3.2]{book}, the function $\mathbf{1}_\phi \cdot \mu_{T'}: L(C')\to \mathcal{M}(U(C'))$  is a kernel. For $T'\in L(C')$,
    $$g(T') = \int_{U(C')} \mathbf{1}_\phi(T',t) \mu_{T'}(dt) = \mathbf{1}_{\phi}\cdot \mu_{T'}(U(C')).$$

    Thus, $g$ is measurable by the definition of the kernel. 
    
    We now show $\hat{\mu}_{L(C')}$-a.e. equivalence in two steps.

First, if a different disintegration kernel with respect to $L(C')$ is chosen,
then by Theorem~\ref{disintegration} the kernels agree
$\hat{\mu}_{L(C')}$-a.e., and so do the resulting functions $g$.

    Next, consider a different choice of all disintegration kernels except for the one with respect to $L(C')$. Denote the characteristic function for this choice as $\mathbf{1}_\phi'$ and the resulting function as $g'$. By the induction hypothesis $\mathbf{1}_\phi'$ is $\mu$-a.e. equivalent to $\mathbf{1}_\phi$. Denote the set where they differ as $\Delta$. Then
    $$0 = \mu(\Delta) = \int_{L(C')} \int_{U(C')} \mathbf{1}_\Delta(T',t) \mu_{T'}(dt)\hat{\mu}_{L(C')}(dT') = $$
    $$ = \int_{L(C')} \int_{U(C')} |\mathbf{1}_\phi(T',t)- \mathbf{1}_\phi'(T',t)| \mu_{T'}(dt)\hat{\mu}_{L(C')}(dT') \ge $$
    $$ \ge \int_{L(C')} | \int_{U(C')} \mathbf{1}_\phi(T',t) \mu_{T'}(dt)- \int_{U(C')} \mathbf{1}_\phi'(T',t) \mu_{T'}(dt)|\hat{\mu}_{L(C')}(dT') =$$
    $$= \int_{L(C')} |g(T') - g'(T')| \hat{\mu}_{L(C')}(dT')$$
    The second equality follows from the definition of the disintegration kernel and the fact that all $\mu_{T'}$ are given by the same kernel. The inequality is due to the reverse triangle inequality.
    
    Integral of nonnegative function $|g'(T')-g(T')|$ with respect to $\hat{\mu}_{L(C')}$ is zero, hence the function is zero $\hat{\mu}_{L(C 
    ')}$-a.e.
    
\end{proof}

\begin{theorem}
    The DTL model-checking problem is undecidable.
\end{theorem}
%We prove the claim by a reduction from the \emph{probabilistic automata emptiness} problem in \Cref{appendix:undecidability}.
\begin{proof}
We reduce from \emph{strict cutpoint emptiness} for probabilistic finite automata (PFA),
which is undecidable (already for cutpoint $1/2$)~\cite{RABIN1963230}.

\textbf{PFA emptiness.}
An instance consists of a finite set of stochastic matrices
$G \subseteq \RR^{d\times d}$, an initial distribution $\pi\in\RR^{d}$, and an
accepting vector $f\in\{0,1\}^{d}$. The question is whether there exist
$m\in\NN$ and $M_0,\dots,M_{m-1}\in G$ such that
\[
\pi^\top M_0 M_1 \cdots M_{m-1} f \;>\; \tfrac12 .
\]

\textbf{Encoding letters.}
Fix a finite set of propositions $AP$ and a surjection $\eta:2^{AP}\twoheadrightarrow G$.
%Without a loss of generality assume that $G$ is enumerated by $2^{AP}$ for a finite set $AP$.
Let $a_f\notin AP$ be a fresh proposition and set $AP_f := AP\cup\{a_f\}$.

\textbf{Markov chain.}
We define a Markov chain
$M=(S,AP_f,l,P,i_M)$ with state space $S:=2^{AP}\times\{1,\dots,d\}$.
The initial distribution is
$i_M(x,j)=\pi_j/|2^{AP}|$ for all $x\in 2^{AP}$ and $j\in\{1,\dots,d\}$.
The labeling is
$$
l(x,j)= \begin{cases}
    x, \text{ if }f_j = 1,\\
    x\cup \{ a_f\}, \text{ otherwise}.
\end{cases}
$$

Transitions are given by, for all $x,y\in 2^{AP}$ and $j,k\in\{1,\dots,d\}$,
\[
P\bigl((x,j),(y,k)\bigr)\ =\ \frac{1}{|2^{AP}|}\cdot \eta(y)(j,k).
\]
Since each $\eta(y)$ is stochastic, $P$ is stochastic as well.

\textbf{DTL formula.}
%Let $\mathsf{Obs}$ be the projection to the $AP$-trace (i.e.\ it forgets $a_f$).
Consider the DTL formula
$$
\phi \ :=\ \PP (\ \LTLeventually (\PP( \,a_f \mid [AP])\ >\ 0.5 )) )\ >\ 0.
$$

\textbf{Correctness.}
Every finite $AP$-prefix (equivalently, every $t\in(2^{AP})^m$) has positive probability
in $M$ by construction. Conditioning on the observation $\mathsf{Obs}=t$ fixes the
letter sequence $t(0),\dots,t(m-1)$ and thus the distribution on the internal index
$j\in\{1,\dots,d\}$ evolves as $\pi^\top \eta(t(0))\cdots \eta(t(m-1))$.
Hence the conditional probability of seeing $a_f$ at the corresponding step equals
$$\pi^\top \eta(t(0))\cdots \eta(t(m-1)) f .$$
Therefore, $M\models\phi$ holds iff there exists a finite word of matrices from $G$
whose acceptance probability exceeds $1/2$, i.e.\ iff the PFA instance is nonempty
above cutpoint $1/2$. Undecidability follows.
\end{proof}

\section{Expressing properties in DTL}

\subsection{Security properties and independence}
We show how DTL can express well-known security properties formulated in terms of independence and conditional independence between different components of an execution trace. Typically, such properties require that conditioning on different high-level behaviors does not change the induced distribution over low-level behaviors.

These properties fall within the linear fragment of DTL, introduced in \Cref{linear}, for which we provide a polynomial-time model-checking algorithm.

\subsubsection{Probabilistic Non-interference}
To define probabilistic non-interference, originally formulated by 
Gray~\cite{130769, 63848}, consider a probabilistic system that communicates with 
agents through high- and low-level channels. Informally, the property requires that 
the low-level trace does not give any information on the distribution of the 
high-level trace, or equivalently, there is no probabilistic information leakage from 
high-level data to low-level data. As defined in~\cite{130769}, probabilistic 
non-interference can be phrased as a conditional-independence requirement. For every 
$n \in \mathbb{N}$, for any high-level input trace $H$ of length $n$, any low-level 
input and output traces $L$ and $O$ of length $n$, and any low-level output event 
$o$ at step $n{+}1$, whenever the joint event $(H, L, O)$ has positive probability, 
the probability of observing $o$ at step $n{+}1$ conditioned on $(H, L, O)$ equals 
the probability conditioned on $(L, O)$.

We demonstrate how to express this property in DTL. Let $M$ be a Markov chain that 
models our probabilistic system, and let the set of atomic propositions be 
$AP = \{h, l, o\}$, where $h$ and $l$ represent the high- and low-level inputs and 
$o$ the low-level output.

We say that $M$ satisfies probabilistic non-interference if
$$ M \models \PP\!\Bigl(
  \LTLglobally \bigl(
    \PP(\LTLnext o \mid [h,l,o])
    =
    \PP(\LTLnext o \mid [l,o])
  \bigr)
\Bigr) = 1.$$

The informal quantification ``for all input traces of length $n$'' is captured by 
the outermost quantification with probability $1$ and the globally operator 
$\LTLglobally$. The innermost equality between probabilistic operators requires that 
$\LTLnext o$ is independent from the history of $h$, conditioning on the history 
of $l$ and $o$.

\begin{remark}
In~\cite{130769}, the system may have an arbitrary number of channels and variables. 
For readability, we use one high-level input, one low-level input, and one output 
proposition. The formulation extends directly to multiple variables by replacing every 
single proposition with a tuple of propositions and updating the conditioning sets 
accordingly.
\end{remark}

\textbf{Relation to other version of probabilistic non-interference.} \label{nonintereference}

In the literature~\cite{abraham2018hyperpctltemporallogicprobabilistic, ClarksonS10, 
1212701}, another version of probabilistic non-interference is used. It requires that 
the distribution of the low-observable trace is the same for every pair of low-level 
equivalent initial states. This version is a weakening of the original probabilistic 
non-interference property: whereas the original requires independence from the 
high-level history at every step of the execution, this version only requires independence from the high-level information of the initial state, leaving open the possibility that high-level information leaks through 
the transition dynamics. For finite state systems, model checking of this weaker 
version can be reduced to a finite number of labeled Markov chain equivalence checks (or probabilistic bisimulation conditions).

In systems with ongoing interaction, the difference is significant. A system may conceal the high-level initial state and thus satisfy the weaker property, while still allowing future low-level observations to depend on the evolving high-level state. In such a system, information about high-level behavior may be leaked dynamically during execution even though the distribution of complete low-level traces is independent of the initial secret. The conditional-independence formulation rules out precisely this form of information flow. In that sense, the original notion of probabilistic non-interference is similar to the possibilistic trace-based version~\cite{6234468, 4556678, clarkson2014temporallogicshyperproperties}. 

\subsubsection{Perfect Indistinguishability}

Perfect indistinguishability~\cite{6769090, Katz2020IntroductionTM,4556678}, also known as perfect secrecy, is commonly
formulated as an indistinguishability requirement for ciphertext distributions,
which in the perfect case reduces to (a form of) independence.

A probabilistic encryption system consists of two components: a key-generation
procedure and an encryption procedure that, given a message and a key, produces
a ciphertext (which can later be decrypted using the key). For simplicity, we
model messages as finite traces over $2^{\{m\}}$ and ciphertexts as finite traces
over $2^{\{c\}}$. We assume that $MD$ receives a message trace and produces a
ciphertext trace, with key generation happening internally.

$MD$ satisfies perfect indistinguishability if for any two messages
$T_1,T_2$ and any ciphertext trace $D$, the probability of outputting $D$ on
input $T_1$ equals the probability of outputting $D$ on input $T_2$.

As above, we model universal quantification over messages by composing $MD$ with
a Markov chain $M_{\mathit{msg}}$ that generates message letters uniformly.
Let
$M := M_{\mathit{msg}} \otimes MD.$

To express perfect indistinguishability in DTL, we need to require that the
message and ciphertext projections are independent. We can do it as follows:
$$M \models \PP\Bigl( \LTLglobally \bigl( \PP( \LTLnext ( \PP( c \mid [m] ) = \PP(c) )  \mid [c] )=1 \bigr) \Bigr) =1 \wedge$$
$$\wedge \PP\Bigl( \LTLglobally \bigl( \PP( \LTLnext ( \PP( m \mid [c] ) = \PP(m) )  \mid [m] )=1 \bigr) \Bigr) =1.$$

 Let us demonstrate that the conjunction above is equivalent to the independence of the projections of an
execution trace onto $(2^{\{m\}})^\omega$ and $(2^{\{c\}})^\omega$.
Let $\mu$ be the probability measure induced by $M$.
Assume the projections are not independent. Then by Lemma 4.6 in~\cite{book} there exist $n,k \in \NN$ and
finite traces $T_m \in (2^{\{m\}})^n$ and $T_c \in (2^{\{c\}})^k$ such that $$\mu((T_m (2^{\{ m\} })^\omega) \times (2^{\{ c\} })^\omega) \cdot \mu((2^{\{ m\} })^\omega \times (T_c(2^{\{ c\} })^\omega)) \neq $$ $$ \neq \mu( (T_m (2^{\{ m\} })^\omega) \times (T_c (2^{\{ c\} })^\omega) ).$$
 Without loss of generality, assume $n\le k$, and assume that for the shorter prefix of $T_c$, it does not hold; otherwise, we can shorten it. In that case, there is at least one $T_{m}' \in (2^{ \{ m\} })^k$, such that
 $T_m' [n] = T_m$, and 
 $$\mu((2^{\{ m\} })^\omega \times (T_c(2^{\{ c\} })^\omega) \mid (2^{\{ m\} })^\omega \times (T_c[k-1](2^{\{ c\} })^\omega)) \neq $$
 $$\neq \mu((2^{\{ m\} })^\omega \times (T_c(2^{\{ c\} })^\omega) \mid ((T_m' (2^{\{ m\} })^\omega) \times (T_c[k-1](2^{\{ c\} })^\omega)),$$
 which fails the first conjunct. If $n\ge k$, the second conjunct fails.

 \subsection{Stochastic environment}\label{environment}

Many stochastic systems consist of several interacting components, one of which models an external stochastic environment. Temporal specifications concern infinite executions, and each complete environment execution (i.e., an infinite trace) naturally induces a conditional probability distribution over the executions of the remaining system. DTL captures these conditional distributions through conditioning on infinite traces, allowing one to express not only the overall probability of satisfying a temporal specification, but also how this probability varies across different environment executions.

This setting naturally occurs in a variety of applications.

\textbf{Communication networks.} A communication network interacting with a stochastic environment modeling the quality of the communication channel and the resulting burst errors \cite{6768434,6769369,articlenetworksurvey}. The Gilbert–Elliott model is widely used to capture burst errors by representing the channel quality as a two-state Markov chain.

\textbf{Hardware systems.} A hardware component interacting with an external stochastic environment modeling radiation-induced transient faults (soft errors), which occur probabilistically \cite{4492840,articlesofterrors}.

\textbf{Embedded systems.} The physical environment of an embedded system can be modeled by a finite-state Markov chain, as in Markov jump linear systems, where the Markov chain represents environmental changes \cite{articlemjls}.

Let the environment be modeled by a Markov chain $M_{env}$ over the set of atomic propositions $AP_{env}$. Let the system be modeled by an MDP $MD$ whose actions are from $2^{AP_{env}}$. Their cascade product $M:= M_{env}\otimes MD$ models the joint behavior of the environment and the system. Let $\mu$ denote the probability measure on infinite traces induced by $M$.   

Let the temporal formula $\phi_{bad}$ specify an undesirable system behavior. DTL can express properties such as 
\textit{With probability at least $0.1$, the environment produces an execution such that, conditioned on this execution, the probability that the system exhibits the undesirable behavior is greater than $0.4$.} This is expressed by the formula
$$M\models \PP\Bigl( \PP\bigl( \phi_{bad} \mid AP_{env}\bigr) > 0.4 \Bigr) > 0.1.$$

In other words, although the system may perform reliably under most environment executions, there exists a substantial subset of environment executions under which its reliability is significantly lower. Such concentration of failures cannot be detected by considering only the overall probability of undesirable behavior. For example, the mean (overall) probability of undesirable behavior over all environment executions may be as low as 0.05, even though there exists a non-negligible subset of environment executions for which the conditional probability of failure exceeds 0.4. Consequently, two systems with the same overall probability of failure may exhibit fundamentally different behavior: one may exhibit a relatively uniform probability of failure across all environment executions, whereas the other may exhibit a high probability of failure concentrated on a significant subset of environment executions. Classical probabilistic temporal logics reason only about the overall probability of satisfying a temporal specification and therefore cannot distinguish between these situations. In contrast, by conditioning on complete environment executions (i.e., infinite traces), DTL enables reasoning not only about the overall probability of satisfying a temporal specification, but also about how this probability depends on environment execution, revealing whether failures are uniformly distributed or concentrated on particular classes of environment executions. In particular, replacing the threshold $0.4$ by $0$ or $1$ yields a property in the qualitative fragment of DTL, for which we present a model-checking algorithm in \Cref{qualitative_section}.

\subsection{Distribution transformers and probabilistic automata}\label{dist_and_prob_aut}

We briefly discuss two verification frameworks for probabilistic systems whose
specifications can be naturally expressed in DTL. The common theme is that,
rather than quantifying existentially over schedulers, one can
reason about a \emph{distribution} over schedulers and ask for the probability
mass of schedulers satisfying a given specification. This modification changes
the original problem and yields a "more probabilistic" counterpart.
For safety and reachability objectives, the two formulations coincide.
For more general $\omega$-regular objectives, however, they may differ, 
and the new one may
even become decidable in cases where the classical existential formulation is
undecidable. Properties that compare probabilities only with 0 and 1 belong to the qualitative fragment of DTL, which we introduce in \Cref{qualitative_section} and show to be decidable.

\subsubsection{MDP as distribution transformers}

Fix an MDP $MD$ and a blind scheduler, i.e.\ an infinite sequence of actions.
Under such a scheduler, the execution of $MD$ induces a sequence of
distributions over states. Properties that refer to this sequence are often
called \emph{distributional properties}. This view received significant attention over the last decade \cite{5600385, 10.1007/978-3-031-37709-9_5, Akshay2018DistributionbasedOF}, with applications in a number of areas, e.g., multi-agent systems \cite{akshay2024certifiedpolicyverificationsynthesis}.

Distributional properties require some linear (specified by linear inequalities) labeling, i.e., functions $\mathcal{P}(S) \to 2^\Sigma$, for the set of states of the MDP $S$ and the set of labels $\Sigma$. \emph{Distributional Model Checking Problem} asks, given an MDP $MD$, a regular set of blind schedulers $\mathcal{A}$, a labeling function $\lambda$, and a B\"uchi automaton $\mathcal{B}$ over $2^\Sigma$, whether there is a scheduler in $\mathcal{A}$, such that it produces the sequence of labels under $\lambda$ which is in $\mathcal{L}(\mathcal{B})$.

Distributional properties can be modeled with DTL, replacing the existential quantification "there is a scheduler" by the probabilistic requirement "there is a set of schedulers of positive measure".

Formally, instead of a regular set of schedulers $\mathcal{A}$, assume a probability distribution over blind schedulers given by a
Markov chain $M_{\mathit{sched}}$, and consider the cascade product
$M := M_{\mathit{sched}} \otimes MD$.
Let $AP_{\mathit{sched}}$ denote the atomic propositions encoding scheduler
choices in $M_{\mathit{sched}}$.
Assume that $M$ is precisely labeled.
For a state $s$ of $MD$, let $\phi_s$ be a (state) formula describing that the $MD$-coordinate of $M$ on this state is $s$.
Then the conditional probability of being in state $s$ given the scheduler
trace is expressible as $\PP(\phi_s \mid [AP_{\mathit{sched}}])$.
Since the logic admits linear combinations of such probabilities, we can encode
any linear labeling $\lambda$, and we can express the LTL property corresponding
to $\mathcal{L}(\mathcal{B})$ (when $\mathcal{L}(\mathcal{B})$ is LTL-expressible).
Let $\phi$ denote the resulting formula.

We can then require that the set of schedulers satisfying $\phi$ has positive
measure. We define the property
$$
M \models \PP( \phi) > 0\bigr.
$$
If $\mathcal{L}(\mathcal{B})$ is a reachability property, this essentially coincides with
pure existential quantification. In general, however, it is stronger, since it
requires a set of schedulers of positive probability mass rather than a single
scheduler.

\subsubsection{Probabilistic Automata over Infinite Words}

Probabilistic automata over infinite words~\cite{pinproceedings, Chatterjee2011DecidablePF} can be viewed as
an MDP in the sense of our modeling framework, together with a B\"uchi, Streett, or Rabin acceptance condition over $2^{AP}$.  A central question studied
in this setting is whether there exists a blind scheduler such that the
probability of producing an infinite trace satisfying a given acceptance
condition meets a prescribed linear
requirement. Typically, this is phrased by comparing the acceptance probability
against thresholds such as $1$, $0$, or an arbitrary $\epsilon \in [0,1]$. It is
known that certain variants, in particular, even those involving comparison against
$0$ or $1$, lead to undecidability~\cite{pinproceedings}.

We can model this setting by assuming a Markov chain $M_{\mathit{aut}}$ with a set of atomic propositions $AP_\mathit{aut}$ such that the action set of $MD$ is $2^{AP_\mathit{aut}}$. Let
$M := M_{\mathit{aut}} \otimes MD$ be the corresponding cascade product. Any acceptance condition over $2^{AP}$ can be modeled by a DTL formula $\phi_{\mathit{acc}}$ without occurrences of probabilistic operators.

Then, for any comparison operator $\star \in \{\le,\ge,=\}$ and any
$\epsilon \in [0,1]$, we can express that the set of schedulers satisfying the comparison has positive probability mass as
$$M\models \PP\Bigl( \PP\bigl( \phi_{acc} \mid AP_{aut}\bigr) \mathrel{\star} \epsilon \Bigr) > 0.$$
Note that this property differs from the original one, since it requires the
existence of a \emph{set} of schedulers (of positive measure), rather than a
single scheduler. Moreover, as we demonstrate in the next section, this property is decidable for $\epsilon \in \{ 0,1\}$, unlike the original one.

\subsection{Related logics} 
The most widely used probabilistic temporal logic, PCTL$^*$~\cite{DBLP:books/daglib/0020348},
combines probabilistic and temporal operators in a branching-time setting.

In the case of a precisely labeled system (when the labeling is injective), DTL subsumes PCTL$^*$. Let $J = [j_1,j_2]$ be an interval. In a precisely labeled system, the
PCTL$^*$ probabilistic operator $\PP_{J}(\phi)$ can be expressed in DTL by
$$
j_1 \le \PP(\phi \mid [AP]) \le j_2,
$$
where conditioning on $[AP]$ corresponds to probability conditioning on the trace up to the current timestamp.

An approach orthogonal to ours to incorporating conditional probabilities into temporal logics is studied in~\cite{DBLP:conf/tacas/BaierKKM14, DBLP:conf/tacas/AndresR08}. The authors consider conditional probabilities of the form \(\PP(\psi \mid \phi)\), where both $\psi$ and $\phi$ are $\omega$-regular events, typically relying on the property $\PP(\psi \mid \phi) = \frac{\PP(\psi \wedge \phi)}{\PP(\phi)}$.

As discussed in Section 1.1 of~\cite{DBLP:conf/tacas/BaierKKM14, DBLP:conf/tacas/AndresR08}, this formalism is also used to express certain security properties involving independence. However, it can only express independence between specific $\omega$-regular outcomes (such as $\LTLglobally l$ and $\LTLfinally h$), and does not allow expressing independence from an arbitrary history, as required for probabilistic non-interference and perfect indistinguishability.

The PHL logic for probabilistic hyperproperties combines probabilistic operators
with explicit quantification over schedulers~\cite{dimitrova2020probabilistichyperpropertiesmarkovdecision}.
To mitigate undecidability, the corresponding model-checking procedure treats
families of schedulers jointly.

Another logic for probabilistic hyperproperties, HyperPCTL~\cite{abraham2018hyperpctltemporallogicprobabilistic}, combines probabilistic reasoning with quantification over states. This allows, for example, expressing a restricted variant of probabilistic non-interference, stating that the distribution of observable outputs is independent of the initial secret state. The difference from the version of probabilistic non-interference considered in our paper is discussed in~\Cref{nonintereference}.

The extension of HyperPCTL called HyperPCTL$^*$~\cite{wang2020statisticalmodelcheckinghyperproperties} allows nesting of probabilistic measurements over paths of the system (sequences of system states) and measurements over tuples of paths. The key difference is that DTL conditions on observable traces (or projections of traces onto selected atomic propositions), whereas HyperPCTL$^*$ reasons about concrete execution paths. Conditioning on traces naturally captures partial observations of an execution and therefore enables the specification of information-flow properties.
For example, the version of probabilistic non-interference considered in~\cite{wang2020statisticalmodelcheckinghyperproperties} coincides with that of HyperPCTL~\cite{abraham2018hyperpctltemporallogicprobabilistic} and requires only independence from the initial secret state.

Logics for possibilistic hyperproperties such as
HyperLTL and HyperCTL$^*$~\cite{clarkson2014temporallogicshyperproperties}
allow nesting and alternation of universal and existential path quantifiers, with temporal operators. Our qualitative fragment can be viewed as a probabilistic counterpart of HyperCTL$^*$, where strict quantifiers are replaced by their "softer" probabilistic counterparts.

There are several program logics for probabilistic conditioning and independence, such as Lilac~\cite{10.1145/3591226}, Bluebell~\cite{Bao_2025}, and Probabilistic Separation Logic~\cite{10.1145/3371123}, as well as probabilistic relational program logics for security properties~\cite{10.1145/1594834.1480894,avanzini2025quantitativeprobabilisticrelationalhoare}. The proof system presented in~\cite{10.1145/3704889} enables reasoning about possibilistic hyperproperties expressed in HyperLTL.

\section{Linear Fragment}\label{linear}
The \emph{linear fragment} of $DTL$ is given by the following grammar:
\begin{align*}
    \kappa &\Coloneqq a \mid \kappa \wedge \kappa \mid \neg \kappa \mid \kappa \LTLuntil \kappa \mid \LTLnext \kappa, \\ 
    \theta &\Coloneqq \PP(\kappa \mid [A_1]) = \PP(\kappa \mid [A_2])
              \mid \LTLglobally\,\theta \mid \LTLnext \theta \mid \PP(\theta \mid [B])=1,\\
    \chi   &\Coloneqq \PP(\theta) = 1 \mid \chi\wedge\chi \mid \neg\chi. 
\end{align*}
Here $A_1,A_2,B\subseteq AP$.

The linear fragment of DTL can express independence of the suffix from some part of the prefix, including properties like probabilistic non-interference and perfect indistinguishability.

In this chapter, we show that model-checking for the linear fragment is in $\mathrm{PTIME}$, using an algorithm inspired by~\cite{doi:10.1137/0221017}. But first, we define an auxiliary construction.

\subsection{Multidynamic Automata}

We define a new type of automata that we use in this section. Intuitively, we use it to prove the generalized notion of equivalence of labeled Markov chains and prove the equivalence of two different dynamics on a graph with respect to some quadratic form.

\begin{definition}
A \emph{multidynamic automaton} is a tuple $D=(n,Q,Q_{\mathrm{in}},Q_{\mathrm{fin}},\delta,z)$,
where $n\in\NN$, $Q$ is a finite set of states, $Q_{\mathrm{in}},Q_{\mathrm{fin}}\subseteq Q$,
$\delta:Q\times Q \rightharpoonup (\RR^{n\times n}\times \RR^{n\times n})$ is a partial function,
and $z:Q_{in}\to\RR^n \times \RR^n$ is a set of pairs of initial vectors.

For $k\ge 1$, a word $p\in Q^k$ is a \emph{path} of $D$ if $p(0)\in Q_{\mathrm{in}}$ and
$\delta(p(i),p(i+1))$ is defined for every $i<k-1$.

For every path $p\in Q^k$, define $val(p)\in \RR^n\times \RR^n$ recursively:
if $k=1$ then $val(p)=z(p(0))$, and if $p'=p\cdot q$ is obtained by appending $q\in Q$ to $p$, then
$$
val(p')=\bigl(\delta(p(k-1),q)_1\, val(p)_1,\ \delta(p(k-1),q)_2\, val(p)_2\bigr).
$$
\end{definition}

Given a matrix $W\in \RR^{n\times n}$, and a multidynamic automaton $D = (n,Q,Q_{in}, Q_{fin},\delta,z)$, we say that $D$ globally satisfies the quadratic form $W$, if for any path $p\in Q^k$ with $p(k-1)\in Q_{fin}$, the equality $val(p)_1^T W val(p)_2 = 0$ holds.

\begin{lemma}\label{multidynamic_lemma}
Given $W\in\RR^{n\times n}$ and a multidynamic automaton
$D=(n,Q,Q_{\mathrm{in}},Q_{\mathrm{fin}},\delta,z)$, deciding whether $D$ globally satisfies $W$
can be done in $\mathrm{PTIME}(|D|\cdot n)$.
\end{lemma}

\begin{proof}
    For $M\in \RR^{n\times n}$, denote as $vec(M) \in \RR^{n^2}$ a vector representation of a matrix $M$. By \cite{IMM2012-03274}, for $vec(AMB) =(B^T \otimes A)vec(M) $, where $\otimes$ denotes the Kronecker product.

    For a path of $D$, $p\in Q^k$, define $lval(p)\in \RR^{n^2}$ recursively. If $k=1,$ define $lval(p) = z(p(0))_2^T\otimes z(p(0))_1^T$. If $p' = pq$ is obtained by appending a state $q\in Q$ to a path $p$, define $lval(p') = (\delta(p(k-1),q)_2^T \otimes \delta(p(k-1),q)_1^T)lval(p)$.
    Note that $lval(p) = val(p)^T_2 \otimes val(p)^T_1$, since $(A\otimes B)(C\otimes D) = AC\otimes BD$. Thus, $val(p)_1^TWval(p)_2 = lval(p)vec(W)$ and it is enough to show that $lval(p) vec(W)=0$ for every path $p$.

    Define $L\subseteq \RR^{Q\times n^2}$, such that $\vec{v}\in L$ if for every $q\in Q_{fin}$, the equality $\vec{v}_q vec(W) =0$ holds.

    Let $E:=\{(q,q')\in Q\times Q \mid \delta(q,q')\ \text{is defined}\}$.
    
    For $(q,q')\in E$, define a linear function $\Delta_{q,q'} : \RR^{Q\times n^2} \to \RR^{Q\times n^2}$, such that $\Delta_{q,q'}(\vec{v})_{q'} = (\delta(q,q')_2^T \otimes \delta(q,q')_1^T)\vec{v}_q$, and for every $q''\neq q',$ $\Delta_{q,q'}(\vec{v})_{q''} = 0^{n^2}$.

    For every $i\in \NN$, define $L_0=L$, and $L_{i+1} = L_i \cap \bigcap_{(q,q')\in E} \Delta_{q,q'}^{-1}(L_i)$. It is a descending chain of linear subspaces of $\RR^{Q\times n^2}$, thus it stabilizes in $|Q|\cdot n^2$ steps, since every decrease implies the decrease of a dimension, thus $\bigcap_{i\in \NN} L_i = L_{|Q|\cdot n^2}$.

    Clearly,  $\bigcap_{i\in \NN} L_i $ is a set of vectors from $L$, for which an application of any finite combination of $\Delta_{q,q'}$ leaves it in $L$. Note that $\Delta_{q,q'}$ models a one-step transition, and $L$ models the set of vector, satisfying $W$, thus $D$ globally satisfies the quadratic form $W$ iff $\vec{v}^z\in  L_{|Q|\cdot n^2}$, where $\vec{v}^z\in \RR^{Q\times n^2}$, with 
    $$\vec{v}^z_q = \begin{cases}
        z(q)_2^T\otimes z(q)_1^T, \ \text{if } q\in Q_{in},\\
        0^{n^2}, \ \text{else}.
    \end{cases}$$

    The last condition can be checked in $\mathrm{PTIME}(|Q|\cdot n)$.
\end{proof}

\subsection{Algorithm}

\begin{theorem} \label{linear_alg}
    The model checking problem for a linear DTL formula $\chi$, and a Markov chain $M$, is decidable in $\mathrm{PTIME}(|M|)$ and in $\mathrm{EXPTIME}(|\chi|)$.
\end{theorem}
\begin{proof}

 By the definition of the semantics, satisfaction of a formula $\kappa$, given by a grammar above, does not depend on a cut or a measure. Let $E^n_\kappa \subseteq (2^{AP})^\omega $ be the set of traces, satisfying $\kappa$ from step $n$ (for any measure and any cut).

By \Cref{constructive_prefix}, for a cut $C:AP\to \NN$, taking only finite values, $n\in \NN$, for $\mu$-a.e. $T$,
 $$\llbracket T,C,n, \PP    (\kappa \mid [A]) \rrbracket_\mu =  \frac{\mu(T_1\in (2^{AP})^\omega \mid T_1\in E^n_\kappa \text{ and } {T_1}_{L(C_A)} = T_{L(C_A)})}{\mu(T_1\in (2^{AP})^\omega \mid {T_1}_{L(C_A)} = T_{L(C_A)})},$$

 where $C_A$ is a cut $C$ updated with $[ A ]$ according to \Cref{eq_s'}.

Let $\mu^s\in \mathcal{P}((2^{AP})^\omega)$ be a measure, induced by a Markov chain $M_s = (S, AP,l, R, i_s)$, with $i_s(s) = 1$.

%Let $E_\kappa\subseteq (2^{AP})^\omega$ the set of traces $T$, such that $T,C_0,0 \models_\mu \kappa$.

As described in \cite{10.1145/210332.210339}, for a given $\kappa$, for every $s$, $\mu^s(E^0_\kappa)$ can be computed in $\mathrm{PTIME}(|M|)$ and in $\mathrm{EXPTIME}(|\kappa|)$.
 
 Denote $M= (S, AP,l, R, i)$. Let $\nu \in \mathcal{P}(S^\omega)$ be a measure induced on the set of paths $S^\omega$ by $M$. $\mu$ is an image of $\nu$ under $l$. Hence,
 $$\mu(T_1\in (2^{AP})^\omega \mid T_1\in E^n_\kappa \text{ and } {T_1}_{L(C_A)} = T_{L(C_A)}) =  $$
 $$ = \nu(\pi \in S^\omega \mid l(\pi)\in E^n_\kappa \text{ and } {l(\pi)}_{L(C_A)} = T_{L(C_A)}) = $$
 $$= \sum_{s\in S} \mu^s(E_\kappa^0) \cdot \nu(\pi\in S^\omega \mid \pi(n) = s \text{ and } {l(\pi)}_{L(C_A)} = T_{L(C_A)}), $$
 %where $\mu^s\in \mathcal{P}((2^{AP})^\omega)$, induced by a Markov chain $M_s = (S, AP,l, R, i_s)$, with $i_s(s) = 1$, and $\nu$ is a measure on $S^\omega$ induced by $M$. Obviously, $\mu$ is an image of $\nu$ under $l$, thus

 Thus,
 $$\mu(T_1\in (2^{AP})^\omega \mid {T_1}_{L(C_A)} = T_{L(C_A)}) = \nu (\pi \in S^\omega \mid  {l(\pi)}_{L(C_A)} = T_{L(C_A)}).$$

 Hence, for $\mu$-a.e. $T$, the equality $T,C,n \models (\PP(\kappa_1 \mid [A_1]) = \PP(\kappa_2 \mid [A_2]))  $ holds iff,
 $$\sum_{s\in S}  \frac{\mu^s(E^0_{\kappa_1}) \cdot \nu(\pi\in S^\omega \mid \pi(n) = s \text{ and } {l(\pi)}_{L(C_{A_1})} =  T_{L(C_{A_1})}))}{\nu (\pi \in S^\omega \mid  {l(\pi)}_{L(C_{A_1})} =  T_{L(C_{A_1})})} = $$
  $$= \sum_{s\in S}  \frac{\mu^s(E^0_{\kappa_2}) \cdot \nu(\pi\in S^\omega \mid \pi(n) = s \text{ and } {l(\pi)}_{L(C_{A_2})} =  T_{L(C_{A_2})}))}{\nu (\pi \in S^\omega \mid  {l(\pi)}_{L(C_{A_2})} =  T_{L(C_{A_2})})},$$
    where $C_{[A_1]}$ and $C_{[A_2]}$ are version of cut $C$ updated according to \Cref{eq_s'} with $[A_1]$ and $[A_2]$, respectively.

    Suppose $\chi = (\PP(\theta) = 1)$ (if $\chi$ is a conjunction of such properties and their negations, each conjunct can be handled separately). By the definition $\theta$ is obtained by repeated application of  $k$ temporal operators $\LTLnext,\LTLglobally$ and some number of probabilistic universal operator $\PP(\dots \mid [B]) = 1$ to $(\PP(\kappa_1 \mid [A_1]) = \PP(\kappa_2 \mid [A_2]))$. Let us enumerate temporal operators from outermost to innermost with $\star_{0}$ being the outermost probabilistic operator, and $\star_{k-1}$ being the innermost one. Moreover, for an operator $\star_i$, denote a subformula to which it is applied as $\theta_i$, hence $\theta$ has a subformula $\star_i\theta_i$. Define $B_i\subseteq AP$ as a set of atomic propositions that appear in universal probabilistic operator inside of $\theta_i$. Formally, for $i<k$ define
    $$B_i = \bigcup (B \mid \theta_i \text{ has a subformula of the form } (\PP(\theta' \mid [B])=1) ).$$
    Additionally define $B_k = \emptyset$.

    Define a multidynamic automaton $D = (|S|, Q, Q_{in},Q_{fin}, \delta, z)$, where $Q=S\times \{ 0,\dots, k\}$, $Q_{fin} = S\times \{k \} $, $Q_{in} = \{(s,0) \mid i(s)>0 \}$.

    For $(s,0)\in Q_{in}$, $y\in \{ 1,2\}$, and $ j<|S|$, define  $$(z(s,0)_y)_j = \begin{cases}
        i(s_j), \text{ if }l(s)_{A_y\cup B_0} = l(s_j)_{A_y\cup B_0},\\
        0, \text{ otherwise.}
    \end{cases}$$
    
    For every $u<k$, and every $s,s'<|S|$, such that $P(s,s') > 0$ define $\delta((s,u),(s',u+1))\in \RR^{|S|\times|S|} \times  \RR^{|S|\times|S|}$ for every $j,j'<|S|$ and $y\in \{ 1,2\}$ as follows 
    $$\delta((s,u),(s',u+1))_y(j,j') = $$ $$ = \begin{cases}
        P(s_j,s_{j'})  \text{ if }l(s_{j'})_{A_y \cup B_u} = l(s')_{A_y \cup B_u}, \\
        0,  \text{ otherwise.}
    \end{cases}$$
    
    Additinally, if $\star_u = \LTLglobally$, define $\delta((s,u),(s',u)) \in \RR^{|S|\times |S|} \times \RR^{|S|\times |S|}$ as follows
    $$\delta((s,u),(s',u))(j,j')  = \delta((s,u),(s',u+1))(j,j').$$

Define, $W'\in \RR^{|S|\times 2}$, $W''\in \RR^{2\times |S|}$, where 
$$W'(j,1) = \mu^{s_j}(E_{\kappa_1}^0), \ \ \ \ W'(j,2) = 1,$$
$$W''(j,1) = -1, \ \ \ \ W''(j,2) = \mu^{s_j}(E_{\kappa_2}^0).$$
Define $W\in \RR^{|S|\times |S|}=W'W''$. Note that for every $v,w\in \RR^{|S|}$,
$$(\sum_{j<|S|} \mu^{s_j}(E_{\kappa_1}^0) v_j)(\sum_{j<| S|}w_j)= (\sum_{j<|S|} \mu^{s_j}(E_{\kappa_2}^0) w_j) (\sum_{j<|S|}v_j)$$
iff $(v)^T W w=0$.

Consider a path $p\in Q^n$ of $D$. It models a path of $M$ (with the first coordinate), with a non-decreasing sequence of numbers (with the second coordinate), which models the temporal operators' choice of the number of step, during the semantical evaluation of $\chi$. Moreover, for $j<|S|$ and $y\in \{1,2\}$, we obtain
$$(val(p)_y)_j = \nu(\pi \in S^\omega \mid \pi(n-1) = s_j \text{ and } \forall i < n. l(\pi(i))_{A_y\cup B_{p(i)_2}} = l(p(i)_1)_{A_y\cup B_{p(i)_2}}  ).$$
Note that for $p\in Q^n$, such that $p(n-1)_2 = k$,
$$(val(p)_y)_j = \nu(\pi\in S^\omega \mid \pi(n-1) = s_j \text{ and } l(\pi)_{L(C_y)} = l(p)_{L(C_y)}),$$
where $C_y$ is a cut that is defined for $a\in AP$ by
$$C_y(a) = \begin{cases}
    n-1, \text{ if } a\in A_y,\\
    max(j<n\mid a\in B_{p(j)_2} \text{ or } j=0), \text{ otherwise.}
\end{cases}$$

Moreover, define a cut $C$, such that $C(a) =  max(j<n\mid a\in B_{p(j)_2})$, hence for $\mu$-a.e. $T$, such that $T_{L(C)} = l(p_1)_{L(C)}$, we obtain  $T,C,n \models (\PP(\kappa_1 \mid [A_1]) = \PP(\kappa_2 \mid [A_2]))$ iff $val(p)_1^T W val(p)_2 = 0$.

%Define a multidynamical automata $D = (|S|, S, S_0, \delta, i)$, where $S_0 = \{ s\mid i(s)>0 \}$, and for every $s,s'\in S$, such that $P(s,s')>0$, define $\delta(s,s') = (M^1_{l(s')_{AP_1}},M^2_{l(s')_{AP_2}})$.

Since every possible sequence of choices of temporal operators is considered, $M \models \chi \text{ iff}$ $D$ globally satisfies the quadratic form $W$. Thus, the Theorem follows from \Cref{multidynamic_lemma}.

\end{proof}

\section{Qualitative Fragment} \label{qualitative_section}
The qualitative fragment restricts the nesting of probabilistic operators such that every operator except the outermost may compare probabilities only to $0$ or $1$.

We define the \textit{qualitative fragment} of DTL by the following grammar.
\begin{align*}
\alpha &  \Coloneqq \psi \leq r \mid \alpha \wedge \alpha \mid \neg \alpha, \\
\psi& \Coloneqq \PP(\gamma) \mid \psi + \psi \mid r\cdot \psi, \\
\gamma&\Coloneqq a \mid \neg a \mid
\gamma \land \gamma \mid \gamma \vee \gamma \mid \LTLnext \gamma  \mid \gamma \LTLuntil \gamma \mid \gamma \LTLr \gamma 
 \mid  \PP  (\gamma \mid [A], B  )  =1 \mid \PP   (\gamma \mid [A], B  )  >0 .
\end{align*}

The qualitative fragment captures some of the properties discussed in section \Cref{environment,dist_and_prob_aut}. In particular, it can express stochastic environment properties in which inner conditional probabilities are compared only to 0 or 1, distributional properties in which the labeling function checks only whether entries of the distribution vector are 0 or 1, as well as properties of probabilistic automata in which the acceptance probability is required to be 0 or 1. As mentioned earlier, the latter is undecidable in the original setting with strict existential quantification over schedulers; we show below that it becomes decidable when moving to our \emph{soft} (probabilistic) quantification.

We call $\PP(\dots) = 1$ and $\PP(\dots) > 0$ \textit{universal} and \textit{existential} probabilistic operators, respectively.

\begin{definition}
The \emph{alternation depth} of a qualitative DTL formula is the maximum number
of alternations between existential and universal probabilistic operators along
any syntactic nesting path. Additionally, each occurrence of $\LTLuntil$ and
$\LTLr$ is counted as two further alternations.
\end{definition}

For $n,m,c\in \NN$ define $g_c(n+1,m) = c^{g_c(n,m)}$, and $g_c(0,m) = m$. We define $NSPACE(g(n,m))$ to be the class of languages that are accepted by a nondeterministic Turing machine that runs in space $O(g_c(n,m))$ for some $c>1$.

Below, we present the main result of this section. It relies on the notion of cascade compatibility, which generalizes properties of models constructed as cascade products of multiple components. A formal definition of this notion is given later in this section.

\begin{theorem} \label{qualitative}
    The model-checking problem for a qualitative DTL formula $\alpha$ of alternation depth $k$, cascade-compatible with a precisely labeled Markov chain $M$, is in $NSPACE(g(k+1,|\alpha|+|M|))$.%and in $NSPACE(g_c(k,|M|))$.
\end{theorem}

We prove this Theorem by presenting an algorithm that, for every subformula $\gamma$, constructs an alternating B\"uchi automaton $B_\gamma$ that models the semantic evaluation of $\gamma$ on an input trace, cut, and timestamp, each encoded by infinite traces. This construction is inspired by the construction from the HyperCTL$^*$ algorithm~\cite{ClarksonFKMRS14}, combined with a reasoning about strongly connected components for the verification with respect to qualitative probabilistic formulas~\cite{DBLP:books/daglib/0020348}.

The automata construction is built recursively on the structure of formula $\gamma$, and uses the Miyano-Hyashi construction for every alternation between universal, existential quantifiers, or the temporal operators $\LTLu$ and $\LTLr$.

\subsection{Cascade property}
We now generalize the property satisfied by the cascade product
$M_1 \otimes M_2$, which will be used to derive a corollary of
\Cref{constructive_disintegration}.

Intuitively, the \emph{cascade property} expresses a conditional independence
assumption. Fix a prefix length $n$ and a prefix $T' \in (2^{AP})^n$.
The property states that, once the projection of $T'$ onto $AP'$ is fixed, the future behavior of the atomic propositions in $AP'$ is independent of
the remaining part of the prefix over $\overline{AP'}$.

\begin{definition}[Cascade property]
A measure $\mu \in \mathcal{P}((2^{AP})^\omega)$ satisfies the
\emph{cascade property} with respect to a set $AP' \subseteq AP$ if for every
$n \in \NN$, every prefix $T' \in (2^{AP})^n$, and every measurable set
$P \subseteq (2^{AP'})^\omega$, the following equality holds:
$$\mu((T'_{AP'} \times (2^{ \overline{AP'}})^n)(P\times(2^{ \overline{AP'}})^\omega)) \cdot \mu( T' (2^{AP})^\omega) = $$ $$= \mu(T'(P\times(2^{ \overline{AP'}})^\omega)) \cdot \mu((T'_{AP'} \times (2^{ \overline{AP'}})^n)(2^{AP})^\omega).$$
\end{definition}

Equivalently, conditioning on the full prefix $T'$ yields the same distribution
over future $AP'$-traces as conditioning only on the projected prefix
$T'_{AP'}$.

\begin{note}
Let $M = (S_1, AP_1, l_1, P, i_1)$ be a Markov chain and let
$MD = (S_2, 2^{AP_1}, AP_2, l_2, R, i_2)$ be a Markov decision process.
Then the measure induced by the cascade product $M \otimes MD$
satisfies the cascade property with respect to $AP_1$.
\end{note}

Given a cut $C \colon AP \to \NN_\omega$ and a number $n \in \NN$, we say that a
property $B \subseteq U(C)$ \emph{depends on the first $n$ steps} if for all
traces $T^1, T^2 \in (2^{AP})^\omega$,
$$
T^1[n] = T^2[n]
\quad\Rightarrow\quad
T^1_{U(C)} \in B \;\text{iff}\; T^2_{U(C)} \in B.
$$

The next result is a corollary of \Cref{constructive_disintegration}, which uses the cascade property to prove that, in certain cases, the limit is reached in finitely many steps.

\begin{corollary}\label{constructive_cascade}
Let $\mu \in \mathcal{P}((2^{AP})^\omega)$ be a measure satisfying the cascade
property with respect to $AP' \subseteq AP$.
Let $C \colon AP \to \NN_\omega$ be a cut and let $N \in \NN$ be such that
$C(a) = \omega$ for all $a \in AP'$ and $C(a) \le N$ for all
$a \in \overline{AP'}$.

Then for $\hat{\mu}_{L(C)}$-almost every $T' \in L(C)$, for every $n \ge N$ and
every property $B \subseteq U(C)$ that depends on the first $n$ steps,
$$
\frac{
  \mu\!\left(
    \{ T'' \in L(C) \mid T''[n] = T'[n] \} \times B
  \right)
}{
  \mu\!\left(
    \{ T'' \in L(C) \mid T''[n] = T'[n] \} \times U(C)
  \right)
}
=
\mu_{T'}(B).
$$
\end{corollary}

\begin{proof}
    For $n' \in \NN$, define
$$
T'_{n'} = \{ T'' \in L(C) \mid T''[n'] = T'[n'] \}.
$$

By Theorem~\ref{constructive_disintegration}, for fixed $n$ and $B$, for
$\hat{\mu}_{L(C)}$-almost every $T' \in L(C)$,
$$
\frac{\mu(T'_{n'} \times B)}{\mu(T'_{n'} \times U(C))}
\xrightarrow[n' \to \infty]{}
\mu_{T'}(B).
$$
    
   By the cascade property, for every $n' \ge n$,
$$
\frac{\mu(T'_{n'} \times B)}{\mu(T'_{n'} \times U(C))}
=
\frac{\mu(T'_n \times B)}{\mu(T'_n \times U(C))}.
$$
Hence, the sequence stabilizes from $n$ onward, and we obtain
$$
\frac{\mu(T'_n \times B)}{\mu(T'_n \times U(C))} = \mu_{T'}(B).
$$

    For fixed $n$ and $B$, this equality holds on a set of $T' \in L(C)$ of
$\hat{\mu}_{L(C)}$-measure one.
Taking a countable intersection over all such $n$ and $B$ preserves measure
one, which concludes the proof.
\end{proof}

We say that a subformula is in the \emph{scope of a total condition on $b$} if it occurs within a probabilistic operator of the form $\PP(\psi \mid [A], B)$ with $b \in B$.
\begin{definition}
    A DTL formula $\phi$ is \textit{cascade-compatible} with a measure $\mu\in \mathcal{P}((2^{AP})^\omega)$ if for every subformula $\phi'$ of $\phi$, $\mu$ satisfies the cascade property with respect to the set 
    $$\{ b\in AP\mid \phi' \text{ is in the scope of a total condition on } b\}.$$
A DTL formula is cascade-compatible with a Markov chain if it is cascade-compatible with a measure induced by that Markov chain.
\end{definition}

Since the cascade property is a conditional-independence requirement, it can be
expressed in a linear fragment of DTL similarly to probabilistic non-interference. Let $M$ be a Markov chain over $AP$ and let $AP' \subseteq AP$. Then $M$ satisfies the
cascade property with respect to $AP'$ iff, for every formula $x  \in 2^{AP'}$, we have
$$M \models \PP\Bigl( \LTLglobally \bigl( \PP(\LTLnext x \mid [AP]) = \PP(\LTLnext x \mid [AP']) \bigr) \Bigr) = 1.$$
Thus, the cascade property can be checked in polynomial time. It follows that the cascade compatibility of a given Markov chain and formula can also be checked in polynomial time.

Hence, \Cref{qualitative} applies directly to systems constructed as cascade products. Moreover, it extends to the more general class of systems that are cascade-compatible with a given formula, and this property can be checked in polynomial time by \Cref{linear_alg}.

\subsection{Automata construction} \label{automata_construction}
Let $M= (S, AP,l, P, i)$ be a precisely labeled Markov chain, and $\gamma$ be a qualitative DTL formula defined by the grammar above. We define an alternating B\"uchi automaton $B_{\gamma, M }$ recursively on the structure of $\gamma$.

The automaton is over the alphabet $\Sigma = 2^{AP} \times \{ 0,\bot,\top\}^{AP} \times \{ 0,\bot\} $. Intuitively, we interpret an infinite trace over $\{ 0,\bot, \top\}$ as an element of $\NN_\omega$. For $n\in \NN$, a trace that starts with $0^{n}\bot $ represents $n$ and a trace that starts $0^n \top $ represents $\omega$. In the same way, an infinite trace over $\{0,\bot\}$ encodes an element of $\NN$.

 This way, an infinite trace over $\{ 0,\bot,\top\}^{AP}$ represents a cut, and an infinite trace over $\{ 0,\bot\}$ represents a timestamp. 
Intuitively, $B_{\gamma, M}$ models the evaluation of $\gamma$ for a given trace, a cut, and a timestamp.

Everywhere in this construction $\sigma\in \Sigma$, $t\in 2^{AP}$, $\vec{v}\in \{ 0,\bot,\top \}^{AP}$, and $b\in \{ 0,\bot \}$.

For $a\in AP$, define $B_{a,M} = (\{ q_0,q_{acc},q_{rej} \}, \Sigma, \delta_{a,M},q_0,q_{acc})$, where 
$$\delta_{a,M}(q_{rej}, \sigma ) = q_{rej}, \ \ \ \ \ \delta_{a,M}(q_{acc},\sigma ) = q_{acc},  $$
$$ \delta_{a,M}(q_0, (t,\vec{v}, 0)) = q_0,$$ $$ \delta_{a,M}(q_0, (t,\vec{v},\bot)) = \begin{cases}
    q_{acc}, \ \ \ \text{if } a\in t, \\
    q_{rej}. \ \ \ \text{else.}
\end{cases}$$

The automaton $B_{\neg a, M}$ is defined in the same way, with the only difference that
$$\delta_{\neg a,M}(q_0, (t,\vec{v},\bot)) = \begin{cases}
    q_{acc}, \ \ \ \text{if } a\notin t, \\
    q_{rej}, \ \ \ \text{else.}
    \end{cases}$$

The automata $B_{\gamma_1\wedge \gamma_2,M}$ and $B_{\gamma_1\vee \gamma_2,M}$ are constructed as an intersection and union of $B_{\gamma_1,M}$ and $B_{\gamma_2,M}$.

The automaton $B_{\LTLnext \gamma, M}$ is defined using the automaton $B_{\gamma,M} = (Q,\Sigma, \delta_\gamma, {q_0}_\gamma, F)$ as 
$B_{\LTLnext \gamma, M} = (Q\times \{ q_0,q_\bot\}, \Sigma, \delta, {q_0}_\gamma \times \{q_0\}, F\times \{q_0, q_\bot\})$, where for $q\in Q$, $i\in \{ 0,\bot\}$,
$$\delta((q,q_i), (t,\vec{v},b)) = \delta_\gamma(q,(t,\vec{v},i)) \times \{ q_b \} .$$

The automaton $B_{\gamma_1 \LTLuntil \gamma_2,M}$ is defined using the automata $B_{\gamma_1,M} = (Q_1,\Sigma, \delta_1, {q_0}_1, F_1)$ and $B_{\gamma_2,M} = (Q_2,\Sigma, \delta_2, {q_0}_2, F_2)$. We assume that $B_{\gamma_1,M}$ and $B_{\gamma_2,M}$ are existential. Otherwise, we turn them into existential using Miyano-Hayashi construction~\cite{Miyano1984AlternatingFA}.

$B_{\gamma_1 \LTLuntil \gamma_2,M} = ((2^{Q_1}\times Q_2\times \{ q_0,q_U\}) \sqcup Q_1\sqcup Q_2, \Sigma, \delta, \{{q_0}_1\} \times {q_0}_2 \times \{ q_0 \},  F_1 \sqcup F_2)$.
Here, for $Q_1'\subseteq Q_1$, $q_1\in Q_1$ and $q_2\in Q_2$ we define 
$$\delta((Q_1', q_2, q_0),(t,\vec{v},0))  =\bigvee_{q_2'\in \delta_2(q_2,(t,\vec{v},0))}(\bigcup_{q_1\in Q_1'} \delta_1(q_1,(t,\vec{v},0)), q_2', q_0),$$
$$\delta((Q_1', q_2, q_0),(t,\vec{v},\bot)) =((\bigvee_{q_1\in Q_1' ;  \ q_1'\in \delta_1(q_1,(t,\vec{v},\bot))} q_1')\wedge $$
$$ \wedge (\bigvee_{q_2'\in \delta_2(q_2,(t,\vec{v},0))} (\bigcup_{q_1\in Q_1'} \delta_1(q_1,(t,\vec{v},0)),q_2',q_U))) \vee (\bigvee_{q_2'\in \delta_2(q_2,(t,\vec{v},\bot))} q_2'),$$
$$\delta((Q_1', q_2, q_U),(t,\vec{v},b)) = \delta((Q_1', q_2, q_0),(t,\vec{v},\bot)),$$
$$\delta(q_1,(t,\vec{v}, b)) = \delta_1(q_1,(t,\vec{v}, b)),$$
$$\delta(q_2,(t,\vec{v}, b)) = \delta_2(q_2,(t,\vec{v}, b)).$$

The automaton $B_{\gamma_1 \LTLr\gamma_2, M}$ is defined as the complement of $B_{(\neg \gamma_1) \LTLu (\neg \gamma_2),M}$, since $\gamma_1 \LTLr\gamma_2=\neg((\neg \gamma_1) \LTLu (\neg \gamma_2))$.

We define automata $B_{\PP(\gamma \mid A,B) = 1,M}$ and $B_{\PP(\gamma \mid A,B) >0,M}$, using the automaton $B_{\gamma, M} $. First, we convert $B_{\gamma, M}$ into a deterministic Rabin automaton $D_{\gamma,M}$.
Next, intuitively, we take a product of $M = (S, AP, l, P, i)$ and $D_{\gamma, M} = (Q_{\gamma}, \Sigma, \delta_\gamma, {q_0}_\gamma, Acc_\gamma)$, interpreting a cut $C$ and a timestamp $n$ as an input. Formally, define an MDP $MD = (S\times Q_\gamma, \{ 0,\bot, \top\}^{AP}\times \{0,\bot \}, AP,l_{MD} , R, i_{MD})$, where $l_{MD}(s,q) = l(s)$, and
$$R((s,q), (\vec{v},b), (s',q')) = \begin{cases}
    P(s,s'), \ \text{if } \delta_\gamma(q,(l(s),\vec{v},b)) = q', \\
    0, \ \text{otherwise}.
\end{cases}$$
$$i_{MD} (s,q) = \begin{cases}
    i(s), \text{ if } q={q_0}_\gamma,\\
    0, \text{ otherwise.}
\end{cases}$$

We call a set $O$ of states of $MD$ \textit{0-strongly connected} if for any $o,o' \in O$, there exists a sequence $o_1,\dots,o_k$, such that $o_1 = o$, $o_k = o'$, and for any $1\le j < k$, $o_j\in O$ and $R(o_j,0^{|AP|+1}, o_{j+1})>0$. A set $O$ is called \textit{0-Bottom strongly connected component} (0-BSCC) if it is a maximal 0-strongly connected set, and for any $o\in O$ and $e\notin O$, $R(o, 0^{|AP|+1},e) = 0$.

We call 0-BSCC $O$ of $MD$ \textit{accepting}, if there exists $(L_i,K_i) \in Acc_\gamma$, such that $O\cap (S\times L_i) = \emptyset$ and $O\cap (S\times K_i)\neq \emptyset$. Otherwise, we call $BSCC$ $O$ \textit{rejecting}. Denote the union of all the accepting $BSCC$ as $ACC$ and the union of all the rejecting $BSCC$ as $REJ$.

Finally we define $B_{\PP(\gamma \mid [A],B) = 1, M} = ((S\times Q_\gamma \times \{ 0,1,\bot,\top\}^{AP} \times \{ 0, \bot\})\cup \{ q_{sink}\}, \Sigma, \delta, q_0,F^{=1})$ as a universal safety automaton, where
$$q_0 = \{ (s,q,0^{AP},0) \mid i_{MD}(s,q) > 0 \}.$$

To define $\delta$, first we need to define the following. Given $\vec{v},\vec{v'}\in \{ 0,1,\bot,\top\}^{AP}$ and $b,b'\in \{0,\bot \}$ define $b'' = 0$, if $b=b'=0$, otherwise $b''=\bot$. Moreover, define $\vec{v''}\in \{ 0,1,\bot,\top\}^{AP}$, such that for every $a\in AP$,
$$\vec{v''}_a = \begin{cases}
    \top, \text{ if } (a\in B) \vee  (\vec{v}_a=\top) \vee (( \vec{v'}_a = \top) \wedge (\vec{v}_a =0)), \text{ else }\\
    0, \text{ if } (\vec{v}_a = \vec{v'}_a =0),  \text{ else }\\
    1, \text{ if } ((\vec{v}_a = 1)\vee (\vec{v'}_a = \bot))\wedge (a\in A) \wedge (b'' = 0), \text{ else} \\
    \bot.
\end{cases}$$
Define $\vec{v''_0} \in \{0,\bot,\top \}^{AP}$  as $\vec{v''}$ with $1$ replaced by $0$.

For $s,s'\in S$, $q,q'\in Q$ and $t\in 2^{AP}$, such that for every $a$, the implication $\vec{v''}_a \in \{ 0,\top\} \Rightarrow l(s)_a = t_a$ holds, define
$$\delta((s,q,\vec{v},b),(t,\vec{v'},b'))=\{ (s',q',\vec{v''},b'') \mid  R((s,q),(\vec{v''_0},b''),(s',q')) > 0\}, $$
Otherwise, define $\delta((s,q,\vec{v},b),(t,\vec{v'},b')) = q_{sink}$. Moreover, $\delta(q_{sink}) = \{ q_{sink}\}$.
%$$F^{=1} = (((S\times Q_\gamma) \setminus REJ) \times \{ \bot,\top\}^{AP}) \cup $$ 
%$$ \cup(S\times Q_\gamma \times (\{ 0,\bot,\top\}^{AP} \setminus \{ \bot, \top \}^{AP} )  ) \cup \{ q_{acc} \}.$$
$$F^{=1} = \overline{ REJ \times \{\bot,\top \}^{AP}\times \{\bot \}}$$
Now let us define $B_{\PP(\gamma \mid [A],B) >0, M} = ((S\times Q_\gamma \times \{ 0,1,\bot,\top\}^{AP} \times \{ 0,\bot\})\cup \{ q_{sink}\}, \Sigma, \delta , {q_0} ,F^{>0} )$ as an existential reachability automaton, with $q_0$ and $\delta$ defined in the same way as in $B_{\PP(\gamma \mid [A],B) =1, M}$, but
$$F^{>0} = ACC \times \{\bot,\top \}^{AP}\times \{\bot \}.$$

\subsection{Correctness and the algorithm}

\begin{lemma} \label{qualitative_lemma}
    Let $n\in \NN$ be a timestemp, $C: AP \to \NN_{\omega}$ be a cut. Let $\pi_n \in 0^{n}\bot \{ 0, \bot\}^\omega$, and  $\pi_C \in( \Pi_{a\in AP; C(a)<\omega} 0^{C(a)} \bot \{0,\bot, \top  \}^\omega) 
    \times (\Pi_{a\in AP; C(a) = \omega} 0^* \top \{0,\bot,\top \}^\omega)$ be traces. Then for $\mu$-a.e. trace $T\in (2^{AP})^\omega$,
    $$T,C,n\models_\mu \gamma \ \text{ iff } \ (T,\pi_C,\pi_n) \in \mathcal{L}(B_{_{\gamma,M}}).$$
\end{lemma}

\begin{proof}
    We prove the statement by structural induction on $\gamma$.
    
    \textbf{Base case.} For $\gamma = a$, or $\gamma = \neg a$  with $a\in AP$, the automaton $B_{\gamma, M}$ waits until the first step where the last component of the input equals $\bot$ (which occurs at position $n$ by the choice of $\pi_n$) and then accepts iff $a\in T(n)$ (respectively, iff $a\notin T(n)$). Hence, the claim holds (for every $T\in (2^{AP})^\omega$) by the definition of the semantics.

    \textbf{Induction step (Boolean operators and $\LTLnext$).} For $\gamma_1\wedge \gamma_2$ or $\gamma_1 \vee \gamma_2$ immediately from the induction hypotheses.

    For $\LTLnext \gamma$, by construction of $B_{\LTLnext \gamma, M}$, for $\mu$-a.e. $T$, $$(T,\pi_C, \pi_n)\in \mathcal{L}(B_{\LTLnext \gamma,M}) \text{ iff } (T,\pi_C, 0\pi_n)\in \mathcal{L}(B_{\gamma,M}).$$

    Thus, the claim follows from the induction hypothesis for $\gamma$.

    \textbf{Induction step ($\LTLuntil$ and $\LTLr$).} Consider $\gamma_1 \LTLu \gamma_2$ and assume that $B_{\gamma_1, M}$ and $B_{\gamma_2, M}$ are existential automata. 

    For $i\in \NN$, denote as $\pi_{n+i}$, the trace $\pi_n$ with prefix of length $n+i$ replaced by $0^{n+i-1}\bot$.
    
    By the induction hypothesis and the definition of the semantics, for $\mu$-a.e. $T$,
    $T,C,n \models_\mu \gamma_1 \LTLu \gamma_2$ iff there exists $i$ such that there exists an accepting run of $B_{\gamma_2,M}$ on $(T,\pi_C,\pi_{n+i})$ and, 
    for every $j<i$, there exists an accepting run of $B_{\gamma_1,M}$ on $(T,\pi_C, \pi_{n+j})$. Note that the induction hypothesis works $\mu$-a.e. for individual $i$ and $j$, but since the number of possible values of $i$ and $j$ is countable, we can conclude that it works $\mu$-a.e. for all $i$ and $j$.

    On the other hand, $(T,\pi_C,\pi_n)\in \mathcal{L}(B_{\gamma_1 \LTLuntil \gamma_2,M})$
    iff there exists an accepting run of the constructed automaton on $(T,\pi_C,\pi_n)$.
    By construction, the run nondeterministically guesses the witness position $i$
    at which $\gamma_2$ should hold, while ensuring that $\gamma_1$ holds at all
    earlier positions. This is equivalent to the existence of $i$ satisfying the
    two conditions above. Hence, the claim follows.

    The case of $\LTLr$ follows from the duality $\gamma_1 \LTLr \gamma_2=\neg((\neg \gamma_1) \LTLu (\neg \gamma_2))$.

    \textbf{Induction step (Probabilistic operators).} Finally, let us prove the induction step for $ \PP ( \gamma \mid [A], B) = 1$. Note that since $D_{\gamma,M}$ is a determinized $B_{\gamma,M}$ the induction hypothesis holds for $D_{\gamma,M}$.

    Recall the definition of $C'$ in \cref{eq_s'}. Let $\pi_{C'} \in( \Pi_{a\in AP; C'(a)<\omega} 0^{C'(a)} \bot \{ 0,\bot,\top \}^* 0^\omega) 
    \times (\Pi_{a\in AP; C'(a) = \omega} 0^* \top \{ 0,\bot,\top \}^* 0 ^\omega)$ and let $\pi'_n \in 0^n\bot \{ 0,\bot\}^* 0^\omega$.

    Note that $(\pi_{C'}, \pi'_n)$ can be modeled deterministically with a finite set of states, thus $MD(\pi_{C'}, \pi'_n)$ can be viewed as a finite state Markov chain. Clearly, $MD(\pi_{C'},\pi'_n)$ induces the same measure $\mu$ on $(2^{AP})^\omega$ as $M$. 

    Any finite path of $MD(\pi_{C'}, \pi'_n)$, that has a positive probability is uniquely determined by its $S$ coordinate corresponding to a state of $M$, since $D_\gamma$ is deterministic and  $(\pi_{C'},\pi_n')$ is modeled deterministically. Since $M$ is precisely labeled, the labeling is injective on the set of finite paths of $MD(\pi_{C'}, \pi'_n)$ with positive probability.
    
    By the Limit Behaviour of Markov Chains~\cite[Theorem~10.27]{DBLP:books/daglib/0020348}, almost every path of $MD(\pi_{C'},\pi'_n)$ visits every state of some BSCC infinitely often. Note that by the construction of $(\pi_{C'},\pi_n')$, BSCCs of $MD(\pi_{C'},\pi'_n)$ are 0-BSCCs of $MD$.

Let \(E \subseteq (2^{AP})^\omega\) be the set of traces \(T_1\) such that either \(T_1\) is not the image, under the labeling function, of a path of \(MD(\pi_{C'}, \pi'_n)\) that takes only positive-probability transitions and visits every state of some BSCC infinitely often, or the equivalence
$$
  (T_1, C', n \models_\mu \gamma
  \;\longleftrightarrow\;
  (T_1, \pi_{C'}, \pi'_n) \in \mathcal{L}(D_{\gamma,M}))
$$
fails.

    By the~\cite[Theorem~10.27]{DBLP:books/daglib/0020348} and the induction hypothesis, $E$ s the union of two measure-zero sets and therefore itself has measure zero.

    Thus,
    $$\int_{L(C')} \int_{U(C')} \mathbf{1}_E(T'
    ,t)  \mu_{T'}(dt) \hat{\mu}'_{L(C')}(dT')  = \mu(E)=0.$$
    Denote as $H\subseteq (2^{AP})^\omega$ the set of traces $T_1$, such that 
    $$T_1, C', n \models_\mu \gamma$$

    Denote as $H'\subseteq (2^{AP})^\omega$ the set of traces $T_1$, such that $T_1$ is an image, under the labeling function, of a path that takes only positive-probability transitions and never visits $REJ$. 

    Hence,
    $$\int_{L(C')} \mid (\int_{U(C')} \mathbf{1}_{H}(T',t)\mu_{T'}(dt)) -(\int_{U(C')} \mathbf{1}_{H'}(T',t)\mu_{T'}(dt)) \mid \hat{\mu}_{L(C')}(dT') \le $$
    $$\le \int_{L(C')} \int_{U(C')}  \mid \mathbf{1}_{H}(T',t) - \mathbf{1}_{H'}(T',t) \mid \mu_{T'}(dt)   \hat{\mu}_{L(C')}(dT') \le$$
    $$\le \int_{L(C')} \int_{U(C')} \mathbf{1}_E(T',t)  \mu_{T'}(dt) \hat{\mu}_{L(C')}(dT') = 0,$$
    where the first inequality is the triangle inequality and the second inequality follows from the definitions $E$, $H$, and $H'$.
    Hence, for $\hat{\mu}_{L(C')}$-a.e. $T'\in L(C')$,
    $$\int_{U(C')} \mathbf{1}_{H}(T',t)\mu_{T'}(dt) = \int_{U(C')} \mathbf{1}_{H'}(T',t)\mu_{T'}(dt).$$
    $$\mu_{T'}(H_{T'}) = \mu_{T'}(H'_{T'}).$$
    
    Thus, for $\mu$-a.e. $T \in (2^{AP})^\omega$,
    $$T,C',n \models_\mu \PP(\gamma \mid [A],B)=1   \text{ iff } $$ $$   \mu_{T_{L(C')}}   ( t\in U(C') \mid (T_{L(C')},t),C',n \models_\mu\gamma ) = 1 \text{ iff }$$
    $$1= \mu_{T_{L(C')}}(H_{T_{L(C')}}) = \mu_{T_{L(C')}}(H'_{T_{L(C')}}) $$
    where the last equality holds $\mu$-a.e. 

    For $N\in \NN$, let $K^N$ denote a set of traces $T_1$, such that $T_1$ is an image under the labeling function, of a path that takes only positive-probability transitions and visits $REJ$ within the first $N$ steps. Clearly,
    % $\bigcup_{N\in \NN} K^N = \overline{H'}$. Hence, 
    $\mu_{T_{L(C')}}(H'(T_{L(C')}))= 1$ iff for every $N\in \NN$, $\mu_{T_{L(C')}}(K^N(T_{L(C')})) = 0$.

    Note that for every $N$ the property $K^N$ depends on the first $N$ steps, since the labeling is injective on finite paths with positive probability.  

    Since $\mu$ satisfies the cascade property, \Cref{constructive_cascade} implies that for $\mu$-a.e. $T$,
    $$\mu_{T_{L(C')}}(K^N_{T_{L(C')}}) = \frac{\mu(T^1\in K^N \mid T^1_{L(C')}[N] = T_{L(C')}[N]) }{\mu(T^1\in (2^{AP})^\omega \mid T^1_{L(C')}[N] = T_{L(C')}[N])}.$$
    
    Thus, $\mu_{T_{L(C')}}   (H') = 1$ iff for any prefix $T_{L(C')}[N]$ there exists no finite trace $p\in (2^{AP})^N$, such that $p_{L(C')} = T_{L(C')}[N]$, $p(2^{AP})^\omega \subseteq K^N$, and $\mu(p(2^{AP})^\omega )>0$. 

    By the construction $B_{\PP(\gamma \mid [A],B) = 1, M}$, receives $\pi_C$ and $\pi_n$ and universally checks for the absence of such $p$, hence the statement follows.

    The proof for $B_{\PP(\gamma \mid [A],B) >0, M}$ is analagous.
\end{proof}

\begin{proof}[proof of \Cref{qualitative}]
    %Define a cut $C_0:AP\to \NN_\omega$, such that $C_0(a) = 0$ for every $a\in AP$.
    Let $M$ be a Markov chain and $\alpha$ a qualitative DTL formula. By definition, $\alpha$ is a Boolean combination of properties of the form $\PP(\gamma)\le r$. Using the construction described in \Cref{automata_construction}, we obtain a B\"uchi automaton $B_{\gamma,M}$ and then translate it into a deterministic Rabin automaton $D_{\gamma,M}$ for each such $\gamma$. We now analyze the complexity of this construction.

    Translating an alternating B\"uchi automaton to a deterministic Rabin
automaton incurs a doubly-exponential blow-up~\cite{Piterman_2007}. In our
construction, this translation is triggered at every alternation between
probabilistic operators and temporal operators. Intuitively, each such
alternation introduces an additional exponent. The complementation of
alternating B\"uchi automata (used to express $\LTLr$) causes a quadratic
blow-up~\cite{10.1145/377978.377993}.

    We next argue that determinization can be avoided when passing from one
existential probabilistic operator to another existential probabilistic
operator. Let $B_\gamma$ be an existential reachability automaton. The formula
$\PP(\gamma \mid [A],B) > 0$ requires that the property expressed by $B_\gamma$
holds with positive probability. Equivalently, there exists a finite path of
nonzero probability along which $B_\gamma$ has an accepting run reaching an
accepting state. Hence, the existential choices can be merged, and no additional
determinization is needed. The argument for the universal-to-universal case is
symmetric.

    Consequently, the size of $D_\gamma$ is bounded by $g(k{+}1,|\alpha|+|M|)$ for an
appropriate tower function $g$, where the additional exponent accounts for the
required determinism.

Given $D_\gamma$, the probability
$\mu\bigl(\mathcal{L}(D_\gamma)\bigr)$ can be computed in time polynomial in
$|D_\gamma|$ and $|M|$~\cite{DBLP:books/daglib/0020348}. It then remains to
evaluate a Boolean combination of linear inequalities, which can be done in time linear in the size of the Boolean formula.

\end{proof}

\section{Conclusion}

We have introduced \emph{Disintegration Temporal Logic (DTL)} for probabilistic hyperproperties.
DTL reasons about probabilities conditioned on (possibly infinite) sequences of
events using the notion of measure disintegration from probability theory.
This allows the logic to express conditional independence and, consequently, hyperproperties
like probabilistic non-interference and perfect indistinguishability.
These hyperproperties belong to the \emph{linear fragment} of DTL, which is decidable in polynomial time.

Moreover, we demonstrated applications of DTL to several classes of probabilistic systems. These include stochastic environments interacting with probabilistic systems, distributional properties of Markov decision processes, and probabilistic automata on infinite words. In each case, DTL reasons about conditional probability distributions induced by finite or infinite observations. For distributional properties and probabilistic automata, DTL replaces strict quantification over blind schedulers with ``soft’’ probabilistic reasoning, yielding decidability in certain settings. Some of these properties belong to the \emph{qualitative fragment}, for which we presented a model-checking algorithm.

%%
%% Bibliography
%%

%% Please use bibtex, 

\bibliography{lipics-v2021-sample-article}

\end{document}